\documentclass[aip, jmp, amsmath, amssymb, reprint, nofootinbib]{revtex4-2}
\usepackage[utf8]{inputenc}
\usepackage[T1]{fontenc}

\usepackage{hyperref}
\usepackage{graphicx}
\usepackage{amsmath, amssymb, amsfonts, amsthm, mathtools}

\usepackage{dsfont}
\usepackage{array}

\usepackage{physics}
\renewcommand{\ket}[1]{\vphantom{#1}\left\lvert\smash{#1}\right\rangle}
\renewcommand{\bra}[1]{\vphantom{#1}\left\langle\smash{#1}\right\rvert}
\renewcommand{\norm}[1]{\braces*{\lVert}{\rVert}{#1}}
\renewcommand{\braket}{\innerproduct*}
\newcommand{\kettbra}{\ketbra*}

\theoremstyle{plain}
\newtheorem{lemma}{Lemma}[]
\newtheorem{thm}[lemma]{Theorem}
\newtheorem{prop}[lemma]{Proposition}
\newtheorem{cor}[lemma]{Corollary}
\theoremstyle{definition}
\newtheorem{definition}{Definition}
\theoremstyle{remark}
\newtheorem{remark}{Remark}
\newtheorem{example}{Example}

\newcommand{\cA}{{\mathcal A}}
\newcommand{\cB}{{\mathcal B}}

\newcommand{\cD}{{\mathcal D}}

\newcommand{\cH}{{\mathcal H}}

\newcommand{\cK}{{\mathcal K}}

\newcommand{\cM}{{\mathcal M}}

\newcommand{\cZ}{{\mathcal Z}}

\newcommand{\bC}{{\mathbb C}}

\newcommand{\bE}{{\mathbb E}}

\newcommand{\E}{{\mathbb E}}

\DeclareMathOperator{\linspan}{span}
\DeclareMathOperator{\gap}{gap}
\DeclareMathOperator{\ran}{ran}
\DeclareMathOperator{\spec}{spec}

\newcommand{\idty}{\ensuremath{\mathds{1}}}

\newcommand{\floor}[1]{\lfloor #1 \rfloor}
\newcommand{\half}[1][1]{\ensuremath{\frac{#1}{2}} }

\usepackage{bm}
\newcommand{\heis}[2]{\bm S_{#1}\cdot \bm S_{#2}}

\begin{document}

\title{A Nonvanishing Spectral Gap for AKLT Models on Generalized Decorated Graphs}

\author{Angelo Lucia}
\affiliation{Departamento de Análisis y Matemática Aplicada, Universidad Complutense de Madrid, 28040 Madrid, Spain}
\affiliation{Instituto de Ciencias Matemáticas, 28049 Madrid, Spain}

\author{Amanda Young}\email{young@ma.tum.de}
\affiliation{Munich Center for Quantum Science and Technology and Zentrum Mathematik, TU M\"{u}nchen 85747 Garching, Germany}

\date{\today}

\begin{abstract}
    We consider the spectral gap question for AKLT models defined on decorated versions of simple, connected graphs $G$. This class of decorated graphs, which are defined by replacing all edges of $G$ with a chain of $n$ sites, in particular includes any decorated multi-dimensional lattice. Using the Tensor Network States (TNS) approach from [H. Abdul-Rahman et. al., Analytic Trends in Mathematical Physics, Contemporary Mathematics, Vol. 741, p. 1 (2020)], we prove that if the decoration parameter is larger than a linear function of the maximal vertex degree, then the decorated model has a nonvanishing spectral gap above the ground state energy.
\end{abstract}

\maketitle

\section{Introduction}

One of the most important classes of quantum spin models in the study of topological phases of matter is the family of the antiferromagnetic, $SU(2)$-invariant quantum spin systems introduced by Affleck, Kennedy, Lieb and Tasaki (AKLT) in Ref.~\onlinecite{AKLT87,AKLT88}. A fundamental quantity in the characterization of quantum phases is the existence or non-existence of a spectral gap above the ground state energy in the thermodynamic limit. In their seminal work, AKLT proved that their one-dimensional, spin-one chain satisfied the characteristic properties of the Haldane phase \cite{haldane:1983}, including a spectral gap of the finite volume Hamiltonians uniform in the system size. AKLT models on higher dimensional lattices were also introduced, and it was further conjectured that if the spatial dimension and coordination number are sufficiently large, then the model would exhibit N\'eel order and, hence, be gapless \cite{AKLT88}. This has been verified analytically for models on Cayley trees with coordination number at least five \cite{AKLT88, Fannes:92b}, and numerical evidence supports the conjecture on three-dimensional lattices \cite{PSA:2009}.

In contrast, the AKLT models on the hexagonal and square lattices were conjectured to be gapped. While it was proved that the AKLT state on the hexagonal lattice does not exhibit N\'eel order in Ref.~\onlinecite{kennedy:1988}, the nonvanishing gap was only recently shown in Ref.~\onlinecite{Lemm_2020,Pomata2019} using a combination of numerical and analytical techniques. The approach in Ref.~\onlinecite{Lemm_2020} uses DMRG with a finite size criterion in the spirit of Knabe \cite{Knabe1988}, while Ref.~\onlinecite{Pomata2019} combines a Lanzcos method with the general theory from Ref.~\onlinecite{decorated-aklt} for proving uniform gaps of quantum spin models with Tensor Network States (TNS) ground states on decorated graphs. The TNS method adapts the one-dimensional finitely correlated state approach from Ref.~\onlinecite{FNW} to a particular class of models defined on decorated graphs defined by replacing each edge of a graph $G$ with a chain of $n$ sites. AKLT ground states on decorated lattices are of interest, e.g., as they have been shown to constitute a universal quantum computation resource \cite{wei:2011, WHR:2014}. The authors of Ref.~\onlinecite{decorated-aklt} applied their theory to show that the AKLT model on the decorated hexagonal lattice was gapped as long as the decoration was sufficiently large. This was then extended in combination with numerical methods to decoration numbers $n\geq 0$ in Ref.~\onlinecite{Pomata_2020}, on the square lattice for decoration parameters $n\geq 2$ in Ref.~\onlinecite{Pomata2019}, and on the 3D diamond lattice and the 2D kagome lattice for $n\ge 1$ in Ref.~\onlinecite{Guo2021}.

Hence, a natural question is whether or not AKLT models on decorated graphs with higher coordination number are gapped. Concretely, we say a quantum spin model on a connected graph $G=(V,E)$ is uniformly gapped if there exists a sequence of finite subgraphs $G_k=(V_k, E_k),$ $k\geq 1$, such that $E_k \subseteq E_{k+1}$, $V_k \subseteq V_{k+1}$ and $\cup_kE_k = E$ for which the associated local Hamiltonians satisfy
\[
\gamma := \inf_k \gap(H_{G_k}) >0,
\]
where $\gap(H_{G_k})>0$ is the difference between the ground state and first excited state energies. The AKLT models have a well-defined infinite volume dynamics in the sense of Ref.~\onlinecite{nachtergaele:2006}. For such quantum spin models, a positive uniform gap implies a nonvanishing gap in the thermodynamic limit, in the sense that the GNS Hamiltonian associated to any weak-* limit of finite-volume ground states of this sequence has a spectral gap bounded below by $\gamma.$ See, e.g. Ref.~\onlinecite{BDN}, for a precise statement and proof.

In this work, we consider any (possibly infinite) simple graph $G=(V,E)$ such that $\Delta(G)=\sup_{v\in V}\deg(v)$ is finite. We show that if all edges of this graph are replaced with a chain of $n$-sites, then the AKLT model on the decorated graph is uniformly gapped as long as $n\geq n(\Delta(G))$ where $n(\Delta(G))$ is a linear function of the maximal vertex degree, see Theorem~\ref{thm:explicit_bound} below. Our proof follows from a slight variation of the analytical framework from Ref.~\onlinecite{decorated-aklt} that produces tighter bounds on the minimal decoration number. The main quantities for bounding the gap using this method depend on the transfer operators associated with the TNS defined by certain quasi one-dimensional subgraphs of the decorated graph. When the maximal vertex degree is small (i.e. 3 or 4), these quantities can be explicitly computed. However, this becomes nontrivial when $\Delta(G)$ is arbitrary. We overcome this challenge by finding the exact singular value decomposition of the transfer operator. 

In a recent work~\cite{2209.01141} with a co-author, we proved that the spectral gap of the decorated AKLT model on the hexagonal lattice is stable when $n\ge 5$, in the sense that the spectral gap remains positive when the model is perturbed by another sufficiently fast decaying interaction. In terms of classification of quantum phases, this means that the model belongs to a stable gapped phase. This stability is a result of proving a condition on the ground states known as \emph{local topological quantum order} (LTQO). It is natural to ask whether this is the case also for the general models considered here, i.e., if it is possible to prove a LTQO condition for sufficiently large decoration number $n$ for any graph $G$. We conjecture that this is in fact the case, although we make no claims on the scaling of the minimal decoration number required. The proof of the result of Ref.~\onlinecite{2209.01141} relies on a representation of the ground states of the AKLT model in terms of a gas of loops~\cite{kennedy:1988}. The LTQO condition is then a consequence of showing that the cluster expansion for the partition function of the loop model converges. While the loop model would be more complicated for general graphs $G$ (for example, it might allow loops to cross each other, something forbidden in trivalent graphs), increasing the decoration number has the effect of decreasing the weight of each segment of the loop, while leaving the other details of the model invariant. So it is reasonable to expect that a large enough decoration would make the expansion convergent even in the more general case.

The organization of this paper is as follows. In Section~\ref{sec:main_result}, we define and state the spectral gap result for the decorated AKLT models, and summarize the modified version of the uniform gap strategy from Ref.~\onlinecite{decorated-aklt}.  A special class of operators, called \emph{matching operators}, are introduced and studied in Section~\ref{sec:matchings}. We use these operators to establish the necessary SVD in Section~\ref{sec:Gap_Proof} and then prove the main result. In the appendices, we provide helpful combinatorial identities, and discuss the differences between the modified TNS approach used in this work and the one proved in Ref.~\onlinecite{decorated-aklt}. In particular, we prove that the modified version used here produces a tighter bound on the minimal decoration needed to guarantee the decorated model is uniformly gapped.

\section{Main result and proof strategy} \label{sec:main_result}

\subsection{The uniform gap for the decorated models}
We consider AKLT models on decorated versions of a (potentially infinite) simple graph $G = (V,E)$ such that
\[ \Delta(G) := \sup_{v\in V} \deg(v) < \infty.\] 
For each $n\geq 1$, the \emph{$n$-decorated graph} $G^{(n)} = (V^{(n)}, E^{(n)})$ is defined by adding $n$ additional vertices to each edge $e\in E$, see Figure~\ref{fig:Dec_graph}. The integer $n$ is called the \emph{decoration parameter}. We will also use $n=0$ to denote the original graph. 

\begin{figure}
    \centering
    \includegraphics[scale=1.1]{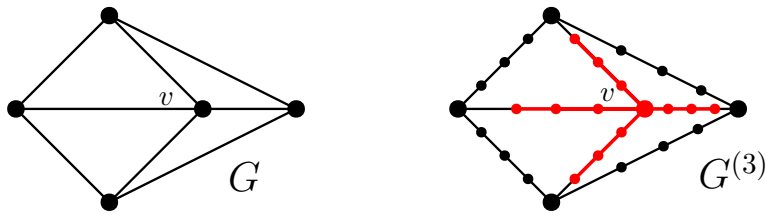}
    \caption{The $n=3$ decorated version of a graph $G$. The subgraph colored in red is $Y_v$.}
    \label{fig:Dec_graph}
\end{figure}

The AKLT model we consider for any $n\geq 0$ is defined as follows. At each vertex $v\in V^{(n)}$ we associate a spin-$\deg(v)/2$ particle represented by the local Hilbert space $\cH_v = \bC^{\deg(v)+1},$ and the interaction for any $e=(v,w) \in E^{(n)}$ is the orthogonal projection \[P^{(z(e)/2)}\in \mathcal{B}(\cH_v\otimes \cH_w)\] 
onto the subspace of maximal total spin, i.e., $z(e):=\deg(v) + \deg(w)$. It is well-known that this defines a frustration-free, nearest-neighbor interaction~\cite{KK89}, and we note that the interaction between any two neighboring decorated sites is simply that of the spin-1 AKLT chain \cite{AKLT88}. 

For any $v\in V$, let $Y_v\subset G^{(n)}$ be the subgraph consisting of the undecorated vertex $v$ and the $n\cdot\deg(v)$ sites decorating the edges incident to $v$, see Figure~\ref{fig:Dec_graph}. Our main result states that as long as the decoration $n$ is sufficiently large, the spectral gap of the finite-volume Hamiltonian
\begin{equation}\label{hamiltonian}
    H_{\Lambda^{(n)}} = \sum_{\substack{\text{edges} \\ (v,w)\in\Lambda^{(n)} }}P_e^{(z(e)/2)}, \qquad \Lambda^{(n)} = \bigcup_{v\in\Lambda}Y_v
\end{equation}
has a positive lower bound independent of $|\Lambda|$ for any finite $\Lambda \subseteq G$. This result, which we now state, depends on the decreasing function
\begin{equation}
    f(d) = 3\cdot\frac{2+(1+\frac{1}{4^d})^{d-1}}{[4-(1+\frac{\sqrt{3}}{2^{d}})^{d-1}]^2}.
\end{equation}

\begin{thm}\label{thm:explicit_bound}
Suppose that $G=(V,E)$ is a simple graph such that $3\leq \Delta(G)<\infty$. If $n\geq n(\Delta(G))$ where
\begin{equation}\label{eq:n_bound}
n(d) = 
\begin{cases}
d & d \leq 4 \\
\frac{\ln(2)}{\ln(3)}d + \frac{\ln(f(d)))}{\ln(3)} & d > 4
\end{cases}
\end{equation}
then there exists $\gamma(\Delta(G),n)>0$ such that
\begin{equation}
    \inf_{\substack{\Lambda\subseteq G : \\ |\Lambda|<\infty}} {\rm gap}(H_{\Lambda^{(n)}}) >\gamma(\Delta(G),n).
\end{equation}
\end{thm}

\begin{remark}
Our result can easily be generalized to the situation where the decoration varies on different edges. Given a bounded function $\bm{n}: E \to \mathds{N}$, consider the decorated graph $G^{(\bm{n})}$ obtained by adding $\bm{n}(e)$ additional vertices to each edge $e \in E$. Then the same arguments used to prove Theorem~\ref{thm:explicit_bound} imply that, for each finite $\Lambda \subset G$, the AKLT Hamiltonian on $\Lambda^{(\bm{n})}$ has a positive spectral gap uniform in $\Lambda$ as long as $\min_{e\in E} \bm{n}(e) \ge n(\Delta(G))$. Here, $\Lambda^{(\pmb{n})}$ is defined as in \eqref{hamiltonian} with $Y_v$ being the subgraph consisting of $v$ and all vertices decorating the edges incident to $v$. The modifications needed to obtain this results are discussed at the end of Section~\ref{sec:TheActualProof}.
\end{remark}

Several other comments regarding Theorem~\ref{thm:explicit_bound} are in order:

\begin{enumerate}
    \item The constraint $\Delta(G)\geq 3$ is not necessary, but the case $\Delta(G)\leq 2$ does not yield any new results. For (undecorated) regular graphs: the case $\Delta(G) = 2$ is the famous AKLT result, while $\Delta(G) =1$ corresponds to an interaction with commuting terms, which is trivially gapped. Moreover, for any graph with $\min_v\deg(v)=1$ and $\Delta(G) = 2$, e.g., a small variation of the one-dimensional finitely correlated states argument from Ref.~\onlinecite{Fannes:92b} would also imply a gap.
    \item In the case that $\Delta(G) = 3,4$, Theorem~\ref{thm:explicit_bound} extends the class of decorated graphs that were studied in Ref.~\onlinecite{decorated-aklt, Pomata2019, Pomata_2020}. Moreover, when comparing with the previous results that only use analytical techniques, the present result either improves or reproduces the lower bound on $n$.
    \item Since $1 \leq f(d) \leq f(5) \approx 1.17851$, Theorem~\ref{thm:explicit_bound} proves a positive uniform gap when $n$ is greater than a linear function of $\Delta(G)$. However, it is unknown if this bound on the minimal decoration is optimal. Since AKLT models on undecorated lattices with large coordination number are expected to exhibit N\'eel order, it would be interesting to determine the minimal decoration needed to guarantee these models are in a gapped phase.
    \item The lower bound on the spectral gap of the model is only a function of the maximal degree $\Delta(G)$. Therefore we can also apply Theorem~\ref{thm:explicit_bound} to show a uniform spectral gap estimate for a sequence of finite graphs $G_k$ having a uniform maximal degree $\Delta$. This allows to prove, for example, uniform bound on the spectral gaps for a sequence of finite volumes with periodic boundary conditions (e.g., $G_k = \mathds{Z}_k^{\times \nu}$ for a fixed $\nu$).
\end{enumerate}

A consequence of Theorem~\ref{thm:explicit_bound} is the following thermodynamic limit result, which is stated with respect to the $C^*$-algebra of (quasi-)local observables
\[
\cA_{G^{(n)}} = \overline{\bigcup_{\substack{\Lambda \subseteq G^{(n)}\\|\Lambda|<\infty}} \cA_{\Lambda}}^{\|\cdot\|}, \qquad \cA_\Lambda = \bigotimes_{\substack{\text{vertices} \\ v\in \Lambda}} \cB(\cH_v).
\]

\begin{cor}
    Let $G$ be as in Theorem~\ref{thm:explicit_bound}, $n\geq n(\Delta(G))$, and $\omega:\cA_{G^{(n)}}\to \bC$ be any weak-$*$ limit of finite-volume ground states 
    \[\omega_m(A) = \langle \psi_m, A \psi_m\rangle \qquad \forall \, A\in \cA_{\Lambda_m^{(n)}}\]
    where $\psi_m\in \ker(H_{\Lambda_m^{(n)}})$ is normalized and $\Lambda_m\uparrow G$. Then, the spectral gap above the ground state of the corresponding GNS Hamiltonian $H_\omega$ satisfies
    \begin{equation}
        \gap(H_\omega) := \sup\{\delta | \spec(H_\omega)\cap (0,\delta) = \emptyset\} \geq \gamma(\Delta(G),n).
    \end{equation}
\end{cor}

As discussed in the introduction, this follows immediately from Theorem~\ref{thm:explicit_bound} by well-known arguments. Theorem~\ref{thm:explicit_bound} is proved using a mild modification of the general framework from Ref.~\onlinecite{decorated-aklt}, which we now review.

\subsection{Reduction to a quasi one-dimensional system} \label{sec:dimension_reduction}

As in Ref.~\onlinecite[Section 2]{decorated-aklt}, rather than estimate the gap of $H_{\Lambda^{(n)}}$ directly we instead consider the gap of the coarse-grained Hamiltonian
\[
\tilde{H}_{\Lambda^{(n)}} = \sum_{v\in \Lambda}P_v, \qquad P_v = \idty - G_v 
\]
where $G_v$ is the orthogonal projection onto the ground state space $\ker(H_{Y_v})$ of $H_{Y_v}.$ Since each edge $e\in \Lambda^{(n)}$ belongs to at most two subvolumes $Y_v$, it is easy to deduce that (analogous to Ref.~\onlinecite[eq. (2.4)]{decorated-aklt})
\begin{equation}\label{eq:equivalence}
    \frac{1}{2}\left(\inf_{v\in V}{\rm gap}(H_{Y_v})\right)\tilde{H}_{\Lambda^{(n)}} \leq H_{\Lambda^{(n)}} \leq \left(\sup_{v\in V}\|H_{Y_v}\|\right) \tilde{H}_{\Lambda^{(n)}}.
\end{equation}
Moreover, ${\rm gap}(H_{Y_v})$ only depends on $\deg(v)$, and so the assumption $\Delta(G)<\infty$ implies that the above infimum is strictly positive. Thus, as the two Hamiltonians have the same ground states, proving a uniform gap for $\tilde{H}_{\Lambda^{(n)}}$ implies a uniform gap of $H_{\Lambda^{(n)}}$.

Estimating the gap of $\tilde{H}_{\Lambda^{(n)}}$ can then be reduced to bounding the quantity
\begin{equation}\label{eq:epsilon_def}
    \epsilon_G(n) = \sup_{(v,w) \in E} \norm{G_v G_w - G_v \wedge G_w}
\end{equation}
where $G_v\wedge G_w$ is the orthogonal projection onto $\ran(G_v)\cap \ran(G_w) = \ker(H_{Y_v\cup Y_w})$ by frustration-freeness. Namely, since $\{P_v, P_w\}\geq 0$ if $(v,w)\notin E$ and 
\begin{align*}
    P_vP_w + P_wP_v & \ge - \norm{P_vP_w - P_v \wedge P_w}(P_v+P_w) \\
    \norm{P_vP_w - P_v \wedge P_w} & = \norm{G_vG_w - G_v \wedge G_w}
\end{align*}
(see Ref.~\onlinecite[Lemma~6.3]{FNW}), $\tilde{H}_{\Lambda^{(n)}}$ satisfies the operator inequality
\[
\tilde{H}_{\Lambda^{(n)}}^2 = \tilde{H}_{\Lambda^{(n)}} + \sum_{w\neq v \in \Lambda} \{P_v, P_w\} \ge \tilde{H}_{\Lambda^{(n)}} - \epsilon_G(n)\sum_{\substack{\text{edges} \\ (v,w)\in \Lambda}} (P_v + P_w).
\]
The final sum above is bounded by $\Delta(G)\tilde{H}_{\Lambda^{(n)}}$ as each vertex $v$ belongs to at most $\Delta(G)$ edges. Hence, ${\rm gap}(\tilde{H}_{\Lambda^{(n)}})\geq 1-\Delta(G)\epsilon_G(n),$ and combining these bounds yields
\begin{equation}
    \operatorname{gap}(H_{\Lambda^{(n)}})  \geq \frac{1}{2} \left(\inf_{v\in V}{\rm gap}(H_{Y_v})\right) \left(1-\Delta(G)\epsilon_G(n)\right) .
\end{equation}
Thus, Theorem~\ref{thm:explicit_bound} immediately follows from showing $\epsilon_G(n) < 1/\Delta(G)$ if $n\ge n(\Delta(G))$. The remainder of this paper is focused on proving this inequality.

\subsection{Transfer operator estimates}

AKLT models are the quintessential class of models with TNS ground states. As shown in Ref.~\onlinecite[Section 3]{decorated-aklt}, the TNS machinery can be used to estimate the norm on the right hand side of \eqref{eq:epsilon_def} for any edge $(v,w)\in E$. Namely, in the situation that the TNS is injective, the norm in \eqref{eq:epsilon_def} is bounded from above by a constant that depends only on the transfer operators associated with various subgraphs of $Y_{v}\cup Y_{w}.$ Our approach is a slight modification of this framework resulting from using a variation of Ref.~\onlinecite[Lemma 3.3]{decorated-aklt} which produces a tighter upper bound on \eqref{eq:epsilon_def}. This variant (namely, Lemma~\ref{lem:IPEstimates}) and that it leads to a better bound are proved in Appendix~\ref{sec:norm-comparison}. In this section, we introduce the necessary transfer operators and state the modified bound on this norm.

\begin{figure}
	\begin{center}
		\includegraphics[scale=.85]{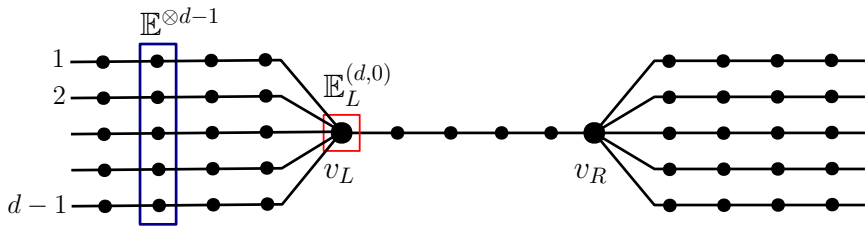}
	\end{center}
	\caption{The region $Y_{v_L}\cup Y_{v_R}$ for two vertices of degree $d=6$ and decoration $n=4$. $X_L$ is the region consisting of $v_L$ and the 5 decorated edges to its left.}
	\label{fig:tensor_vertex}
\end{figure}

\subsubsection{The transfer operators}\label{sec:transfer_ops}
Label $v_L$ and $v_R$ as the `left' and `right' vertex, respectively, associated with an edge $(v_L,v_R)\in E$, and decompose
\[
Y_{v_L} \cup Y_{v_R} = X_L \cup C_n \cup X_R
\]
where $C_n= Y_{v_L}\cap Y_{v_R}$ is the chain of $n$ sites decorating the edge $(v_L,v_R)$, and  $X_\# = Y_{v_\#} \setminus C_n$  for $\#\in \{L,R\}$, see Figure~\ref{fig:tensor_vertex}. The desired transfer operators are those associated with these three regions.

The transfer operator for the center region $C_n$ is the $n$-fold composition $\bE^{n}$ where $\bE:M_2(\bC)\to M_2(\bC)$ is the (well known) single vertex spin-1 AKLT transfer operator
\begin{equation}
    \label{eq:Spin1_transfer_op}
    \bE= \ketbra{\idty}{\rho} -\frac{1}{3} \sum_{U=X,Y,Z} \ketbra*{\sigma^U}{S^U}.
\end{equation} 
Here, $\rho=\idty/2$, $S^U=\sigma^U/2$, and the bra-ket notation is with respect to the Hilbert-Schmidt inner product. The convergence of $\bE^n$ to its fixed point is known. Namely,
\begin{equation}
    \label{eq:1d_convergence}
    a(n) = \norm{ \bE^n -\ketbra{\idty}{\rho}} = \frac{1}{3^n}.
\end{equation}

Now, let $d_{\#}=\deg(v_\#)$ for $\# \in \{L,R\}$, and denote by
\[
\bE_L^{(d_L,n)}: M_{2}(\bC) \to M_{2}(\bC)^{\otimes d_L-1}, \quad \bE_R^{(d_R,n)}:M_{2}(\bC)^{\otimes d_R-1} \to M_{2}(\bC),
\] 
the transfer operators associated to $X_L$ and $X_R$, respectively. These are the composition of the undecorated transfer operator and $n$ copies of $\bE^{\otimes d_{\#}-1}$:
\begin{equation}\label{eq:0_to_n_relation}
    \bE^{(d_L,n)}_L = (\bE^{\otimes d_L-1})^n \circ \bE^{(d_L,0)}_L, \quad  \bE^{(d_R,n)}_R =  \bE^{(d_R,0)}_R\circ (\bE^{\otimes d_R-1})^n.
\end{equation}
See Figure~\ref{fig:tensor_vertex} for a visualization. As $\cH_{v_L} = \bC^{d+1}$ when $d_L=d$, by definition
\begin{equation}\label{eq:undecorated_transfer_op}
    \bE^{(d,0)}_L(B) = \sum_{k=0}^d (W_k^d)^* B W_k^d,
\end{equation}
where $\{W_k^d : 0 \leq k \leq d\}$ are the tensors associated with $v_L$ in the TNS representation of the ground states. We construct these using the valence bond state formalism for the AKLT ground states where we use the convention that the edge(s) to the right-side of a vertex are projected into the singlet state \cite{fannes:1989}. 

Identify $\cH_v$ with the symmetric subspace of $d$ spin-1/2 particles, and let $\phi_k^d\in(\bC^2)^{\otimes d}$ be the normalized symmetric vector with $k$ up spins. These satisfy the recursive formula
\begin{equation}\label{eq:symmetric-vector-recursion}
\phi_k^d = \left(\frac{d-k}{d}\right)^{1/2}\ket{\downarrow}\otimes\ket{\phi_k^{d-1}}+\left(\frac{k}{d}\right)^{1/2}\ket{\uparrow}\otimes\ket{\phi_{k-1}^{d-1}}, \quad 0 \leq k \leq d.
\end{equation}
Then, denoting by $K=\kettbra{\uparrow}{\downarrow}-\kettbra{\downarrow}{\uparrow}$ the singlet tensor, one can take $W_k^d = K V_k^d$ where
\begin{equation}
	V_k^d =
	\left(\frac{d-k}{d}\right)^{1/2}\kettbra{\downarrow}{\phi_{k}^{d-1}}+\left(\frac{k}{d}\right)^{1/2}\kettbra{\uparrow}{\phi_{k-1}^{d-1}}.
\end{equation}
We note that when $d=2$, this produces a scaled version of the 1-dimensional AKLT transfer operator from \eqref{eq:Spin1_transfer_op}, namely, $\E_L^{(2,0)}= \frac{3}{2}\bE$. Such a scaling is inconsequential, but convenient for producing consistent formulas.

The same procedure works for defining $\bE_R^{(d,0)}$ via the tensors $\tilde{W}_k^d$ associated with $v_R$. This yields
$\tilde{W}_k^d := K^{\otimes d-1}(V_k^d)^*=(-1)^{k-1}(W_{d-k}^d)^*,$
which combined with \eqref{eq:Spin1_transfer_op}-\eqref{eq:undecorated_transfer_op} implies
\begin{equation}\label{eq:ER}
    \bE_{R}^{(d,n)} = (\bE_{L}^{(d,n)})^* \qquad \forall \, n\geq 0,
\end{equation}
where the adjoint is with respect to the Hilbert-Schmidt inner product.

\subsubsection{Bounding $\epsilon_G(n)$ via transfer operator quantities}

The desired estimate on $\epsilon_G(n)$ for $G=(V,E)$ is a consequence of bounding
\[
\epsilon_{(v_L,v_R)}(n):= \|G_{v_L}G_{v_R}-G_{v_L}\wedge G_{v_R}\| \qquad \forall\, (v_L,v_R)\in E
\]
by a constant that only depends on three types of quantities. Namely, the convergence rate $a(n)$ from \eqref{eq:1d_convergence}, the minimal eigenvalues $q_L^{(d,n)}$ and $q_R^{(d,n)}$, respectively, of
\begin{equation}\label{eq:QL-definition}
    Q_L^{(d,n)} := \bE^{(d,n)}_L(\idty), \qquad  Q_R^{(d,n)} :=  (\bE^{(d,n)}_R)^*(\rho),
\end{equation}
and the Schatten-$\infty$ norm of $\bE_L^{(d,n)}$ and $\bE_R^{(d,n)}$. Here, we recall that with respect to the Hilbert-Schmidt norm $\norm{\cdot}_2$,
\begin{equation}\label{eq:Shatten_norm}
  \|\bE_L^{(d,n)}\|_\infty  = \sup_{0\neq A\in M_{2}(\bC)}\frac{\|\bE_L^{(d,n)}(A)\|_2}{\|A\|_2} = \max \left\{ \sigma \mid \sigma \text{ singular value of } \bE^{(d,n)}_L \right\},
\end{equation}
and $\|\bE_R^{(d,n)}\|_\infty =\|\bE_L^{(d,n)}\|_\infty$ by \eqref{eq:ER}. Given these values, the following is the mild variation of Ref.~\onlinecite[Proposition 3.6]{decorated-aklt} that results from replacing Ref.~\onlinecite[Lemma~3.3]{decorated-aklt} with Lemma~\ref{lem:IPEstimates} (see Appendix~\ref{sec:norm-comparison}) everywhere in the former work.

\begin{prop}\label{prop:espilon-estimate}
Let $(v_L,v_R)$ be an edge of a simple graph $G$ between vertices with degrees $d_L$ and $d_R$, respectively, and define
\begin{equation}\label{eq:b_defs}
b_L(d,n) = 4 a(n) \frac{\|\bE_L^{(d,n)}\|_\infty}{q_L^{(d,n)}}, \quad b_R(d,n) = 2a(n) \frac{\|\bE_R^{(d,n)}\|_\infty}{q_R^{(d,n)}}.
\end{equation}
If $\max\left\{b_L(d_L,n),\, b_R(d_R,n),\, \tfrac{b_L(d_L,n)b_R(d_R,n)}{4a(n)}\right\} <1$, then $\epsilon_{(v_L,v_R)}(n) \leq \delta_{d_L,d_R}(n)$ where
\begin{equation}\label{eq:simplified_epsilon}
\delta_{d_L,d_R}(n) = 4 a(n) \left(\frac{1}{\sqrt{(1-b_L(d_L,n))(1-b_R(d_R,n))}} + \frac{4 a(n) + b_L(d_L,n)b_R(d_R,n)}{(1-b_L(d_L, n))(1-b_R(d_R,n))}\right).
\end{equation}
\end{prop}

We note that the maximum constraint in the proposition is sufficient to prove that the TNS representation of the ground states associated to $Y_{v_L}, \, Y_{v_R}$ and $Y_{v_L}\cup Y_{v_R}$ are all injective, see, e.g., Ref.~\onlinecite[Corollary~3.4]{decorated-aklt}, which is necessary for their approach. 

Since $\rho = \idty/2$, it follows from \eqref{eq:ER} that 
\begin{equation}\label{eq:b_reduction}
b_L(d,n)=b_R(d,n):=b(d,n).
\end{equation} 
Hence, the strategy for proving Theorem~\ref{thm:explicit_bound} is to bound $b(d,n)$ so that by \eqref{eq:epsilon_def} and Proposition~\ref{prop:espilon-estimate},
\begin{equation}\label{eq:strategy}
    \epsilon_G(n) \leq \sup_{(v_L,v_R)\in E} \delta_{d_L,d_R}(n) < \frac{1}{\Delta(G)} \qquad \forall \, n \geq n(\Delta(G)).
\end{equation}
In Lemma~\ref{lem:SVD_Decomp}, we determine the SVD of $\bE_L^{(d,n)}$, from which both $Q_L^{(d,n)}$ and $\|\bE_L^{(d,n)}\|_\infty$ are easily calculated. Bounds on the quantities defining $b(d,n)$ that will imply \eqref{eq:strategy} are then proved in Lemma~\ref{lem:quantity_bounds}. The final details establishing \eqref{eq:strategy} are the content of the proof of Theorem~\ref{thm:explicit_bound} in Section~\ref{sec:TheActualProof}.

While this is discussed more thoroughly in Appendix~\ref{sec:norm-comparison}, we point out that the main difference between the two approaches comes from using the Schatten-$\infty$ norm in \eqref{eq:b_defs} rather than the norm induced by the operator norm, i.e.
\begin{equation}\label{eq:op_norm}
  \|\bE_{\#}^{(d,n)}\|_{\rm op} = \sup_{A\neq 0}\frac{\|\bE_{\#}^{(d,n)}(A)\|}{\|A\|}.  
\end{equation}
Two additional comments are in order. First, $a(n)$ is invariant under this change of norms as an elementary, but tedious, calculation shows
\[\|\bE^n-\ketbra{\idty}{\rho}\|_\infty = \|\bE^n-\ketbra{\idty}{\rho}\|_{\rm op}.\]
Second, the complete positivity of the transfer operator guarantees that \[\|\bE_\#^{(d,n)}\|_{\rm op} = \|\bE_\#^{(d,n)}(\idty)\|,\qquad \#\in\{L,R\}.\] Contrary to the statement just before Ref.~\onlinecite[Equation 4.18]{decorated-aklt}, one can show that $\|\bE_L^{(d,n)}(\idty)\| \neq \|\bE_R^{(d,n)}(\idty)\|,$ see, e.g., Appendix~\ref{sec:norm-comparison}. Correcting this small error in the proof of Ref.~\onlinecite[Theorem~2.2]{decorated-aklt} yields the existence of a gap for $n\geq 4$ instead of the claimed $n\geq 3$. However, as Proposition~\ref{prop:bound_comparison} illustrates, Proposition~\ref{prop:espilon-estimate} gives a tighter bound on the minimal decoration needed to guarantee a spectral gap than the method from Ref.~\onlinecite{decorated-aklt}. Using the present method recovers the gap result for $n=3.$

\section{Matchings operators}\label{sec:matchings}

The goal of this section is to introduce a special collection of eigenvectors of 
\[\bE^{\otimes d-1}:M_2(\bC^2)^{\otimes d-1} \to M_2(\bC^{2})^{\otimes d-1},\] 
which we refer to as \emph{matching operators}. These will aid us in calculating the SVD of the decorated transfer operator $\bE_L^{(d, n)}$ from the SVD of the undecorated transfer operator $\bE_L^{(d, 0)}$.

In addition to introducing the matching operators, we show that they form an orthogonal basis (with respect to the Hilbert-Schmidt inner product) for the commutative algebra generated by the spectral projections of the Casimir operator. This follows from proving that the matching operators satisfy certain recursion relations. These relations will also provide a path for explicitly writing individual spectral projections of the Casimir operator in terms of matching operators, which is the key for determining the desired SVD.

\subsection{Basic notions of matchings and matching operators}

To begin, fix $m\geq 2$ and recall that the spectral decomposition of the Casimir operator $C^{(m)} = (\sum_{i=1}^m\mathbf{S}_i)^2$ acting on $(\bC^2)^{\otimes m}$ is
\begin{equation}
    \label{eq:Casimir}
    C^{(m)} = \sum_{j\in J_m} j(j+1)Q^{(m,j)}
\end{equation}
where $Q^{(m,j)}\in M_2(\bC^2)^{\otimes m}$ is the orthogonal projection onto the direct sum of all subspaces of total spin $j$, and 
\[
J_m = \left\{j_0 + k : 0 \leq k \leq \floor{\tfrac{m}{2}}\right\},
\]
with $j_0 = 0$ if $m$ is even and $j_0 = 1/2$ otherwise. Moreover, recall that $Q^{(m,j)}= P_{\rm sym}^{(m)}$ is the projection onto the symmetric subspace  when $j= \floor{\tfrac{m}{2}}+j_0$.

The commutative algebra generated by the Casimir operator is then
\begin{equation}
\cZ^{(m)} = \linspan\{ Q^{(m,j)} \mid j \in J_m\},
\end{equation}
which has dimension $\abs{J_m} = \floor{\frac{m}{2}} +1$. The goal is to show that the set of matching operators forms an orthogonal basis for $\mathcal{Z}^{(m)}$ which are eigenvalues of $\bE^{\otimes m}$. To introduce these operators, we first define the notion of a \emph{matching}.

\subsubsection{Definition and properties of matchings}
\begin{definition} For any $1\leq r \leq \floor{\frac{m}{2}}$, an \emph{$r$-matching} of $[m] := \{1,\dots, m\}$ is a collection of $r$ unordered\footnote{Meaning that $(a,b)$ is the same pair as $(b,a)$} pairs \[p = \{(a_1,b_1),\dots,(a_r,b_r)\}, \quad a_i, b_i\in [m] \] 
which are distinct in the sense that $\cup p:=\{a_i, b_i : 1\leq i \leq r\}$  has $2r$ elements. The set of all $r$-matchings is denoted by $\mathcal M^{m}_{r}$. For consistency,  $\mathcal{M}^{m}_0$ is the set consisting of the \emph{empty matching}, $p=\{\}$.
\end{definition}

\begin{example}
$\cM^4_{1}= \big\{\{(i,j)\}: 1\leq i<j\leq 4\big\}$ has six elements, and $\cM^4_2$ consists of
\[\{ (1,2), (3,4) \}, \quad \{ (1,3), (2,4) \} \quad \text{and} \quad \{ (1,4), (2,3) \}.
\]
\end{example}

The number of $r$-matchings is easily calculated with the multinomial coefficient:
\begin{equation}\label{eq:matching_count}
    \abs{\mathcal M^{m}_r} = \frac{1}{r!} \binom{m}{2, \ldots, 2, m-2r} = \frac{m!}{2^rr! (m-2r)!} = \binom{m}{2r} (2r-1)!!
\end{equation}
Here, $n!!$ denotes the double factorial, i.e., the product of all integers from 1 to $n$ having the same parity as $n$. By convention $0!! = (-1)!! = 1$, and so \eqref{eq:matching_count} also holds when $r=0$.

For the main recursion result, it will also be important to know how many matchings $p\in\cM_r^m$ contain either both, one, or none of the elements from a single pair $(i,j)$. As such, we introduce the following (possibly empty) sets, which form a partition of $\cM_r^m.$

\begin{definition} \label{def:matching_partition}
Fix $m\geq 2$. For any $0 \le r \le \floor{\half[m]}$ and distinct pair $i,j \in [m]$, set
  \begin{align*}
    A_r(i,j) &= \left\{ p \in \mathcal M_r^{m} \mid (i,j) \in p\right\},\\
    B_r(i, j) &= \left\{ p \in \mathcal M_r^{m} \mid i \in \cup p, j \not \in \cup p \;\,\text{or}\;\, j \in \cup p, i \not \in \cup p  \right\},\\
    C_r(i, j) &= \left\{ p \in \mathcal M_r^{m} \mid i, j \in \cup p, (i,j)\not \in p \right\},\\
    D_r(i, j) &= \left\{ p \in \mathcal M_r^{m} \mid i, j\not \in \cup p \right\}.
  \end{align*}
\end{definition}
Clearly, the cardinality of these sets is independent of the choice of $(i,j)$, and only depends on $m$ and $r$.
To shorten notation, we will denote by $\abs{A_r}$ the cardinality of the sets $A_r(i,j)$, and similarly for the others.

\begin{lemma}\label{lemma:cardinality} Fix $m>2$. The sets from Definition~\ref{def:matching_partition} have the following cardinalities for any $r$ with the convention that $\abs{\mathcal{M}_{r}^m}=0$ if $m< 0$ or $r<0$:
  \begin{align}
    \abs{A_r} &= \abs{D_{r-1}} = \abs{\mathcal{M}^{m-2}_{r-1}}\\
    \abs{B_r} &= 2(m-2)\abs{\mathcal{M}^{m-3}_{r-1}}\\
    \abs{C_r} &= (m-2)(m-3) \abs{\mathcal{M}^{m-4}_{r-2}}.
  \end{align}
  
\end{lemma}
\begin{proof} For each set $X_r(i,j)$ from Definition~\ref{def:matching_partition}, the result follows from constructing an $n$-to-1 mapping between $X_r(i,j)$ and the appropriate set of matchings.

\begin{enumerate}
    \item A bijection between $A_r(i,j)$ and $D_{r-1}(i,j)$ is given by
  \[ p \in A_r(i,j) \mapsto p\setminus\{(i,j)\} \in D_{r-1}(i,j),\]
while a bijection between $D_{r-1}(i,j)$ and $\mathcal{M}^{m-2}_{r-1}$ results from recognizing $D_{r-1}(i,j)$ as the set of all $r-1$ pairings on the set $[m]\setminus\{i,j\}$.
    \item For $p\in B_{r}(i,j)$ define the mapping $p \mapsto p'$ by decomposing
\[p = \{(i,x)\}\cup p' \quad \text{or}\quad p = \{(j,x)\}\cup p'\] 
for some $x\in [m]\setminus \{i,j\}.$ Moreover, after fixing $x$, the image $p'$ can be any $r-1$ matching on the set $[m]\setminus \{i,j,x\}$. As there are $m-2$ choices for $x$ and two possible pairings, i.e. $(i,x)$ and $(j,x)$, the result follows.
    \item For $p\in C_r(i,j)$, the mapping $p\mapsto p'$ is defined by writing
\[
p = \{(x,i),\,(y,j)\}\cup p'
\]
where $x,y\in[m]\setminus\{i,j\}$. As before, $p'$ is an $r-2$-matching on $[m]\setminus \{x,y,i,j\}$. Since there are $(m-2)(m-3)$ distinct choices for $x,$ and $y$ the result follows.
\end{enumerate}

\end{proof}

\subsubsection{Definition and properties of matching operators}

We now introduce the operators of interest: the \emph{matching operators.}

\begin{definition} Fix $m\geq 2$. Then, the \emph{$r$-matching operator}  $M^{(m)}_r\in \cB ( (\mathds{C}^2)^{\otimes m})$  is
\begin{equation}  \label{eq:matching_ops}
  M^{(m)}_r := \sum_{p \in \mathcal{M}^m_r} S_p, \qquad S_p: = \prod_{(i,j)\in p} \heis{i}{j}
\end{equation}
for all $1 \leq r \leq \floor{\tfrac{m}{2}}$ and $M^{(m)}_0 = S_{\emptyset} = \idty$. 
\end{definition}

Note that since the pairs are disjoint, the order of the product in $S_p$ does not matter. Moreover, each matching operator is clearly nonzero since
\begin{equation}\label{eq:M-nonzero}
    S_p\ket{\uparrow\ldots\uparrow} = \frac{1}{4^r}\ket{\uparrow\ldots\uparrow}, \qquad \forall\, p\in \cM_r^m
\end{equation}
Several other important observations are in order.
\begin{enumerate}
\item $\bE^{\otimes m}(M_r^{(m)})=3^{-2r}M_r^{(m)}$ since $\bE(\idty) = \idty$ and $\bE\otimes\bE (\heis{\phantom{i}}{\phantom{i}}) = \frac{1}{9}(\heis{\phantom{i}}{\phantom{i}}).$

\item $M^{(m)}_1$ is related to the Casimir operator $C^{(m)}$ via
    \begin{equation}\label{eq:Casimir_oneMatching}
      C^{(m)} = \frac{3m}{4} \idty + 2 M^{(m)}_1.
    \end{equation}  
\item Every matching operator $M^{(m)}_r$ is Hermitian as well as $SU(2)$ and permutation symmetric (since $\cM_r^m$ is permutation invariant). Therefore, $M_r^{(m)}\in \mathcal{Z}^{(m)}$ by Schur-Weyl duality. However, we provide an alternative way of verifying this inclusion. Lemma~\ref{lemma:product-m1-mr} establishes that
\[M_1^{(m)}\cdot M_r^{(m)}= c_r M_{r-1}^{(m)} + a_rM_{r}^{(m)} + b_rM_{r+1}^{(m)} \quad \forall r\geq 0,\]
for appropriate coefficients $a_r,b_r,c_r$. As $M_0^{(m)} = \idty$, repeatedly applying this relation shows that every matching operator can be written as a polynomial of $M_1^{(m)}$, and hence $C^{(m)}$ by \eqref{eq:Casimir_oneMatching}.
\end{enumerate}

\begin{prop}\label{prop:matchings-are-basis}
The set of matching operators $\{ M^{(m)}_r \}_{r=0}^{\lfloor \frac{m}{2} \rfloor }$ forms a Hilbert-Schmidt orthogonal basis of  $\mathcal{Z}^{(m)}$.
\end{prop}
\begin{proof} Since each $M^{(m)}_r\in \cZ^{(m)}$ is nonzero by \eqref{eq:M-nonzero}, it is sufficient to show the set of matching operators is orthogonal, as the number of matching operators is the same as the dimension of $\mathcal{Z}^{(m)}$.

Let $r_1\neq r_2$. Since each matching operator is Hermitian, it trivially follows that
 \begin{equation}\label{eq:hs-product}
  \tr[ (M_{r_1}^{(m)})^* (M_{r_2}^{(m)}) ] = \tr[ (M_{r_1}^{(m)}) (M_{r_2}^{(m)}) ] = \sum_{ \substack{p \in \mathcal{M}_{r_1}^{m} \\ q \in \mathcal{M}_{r_2}^m}} \tr(S_p S_q).
\end{equation}
Considering \eqref{eq:matching_ops}, on sees that $\tr(S_p S_q) = 0$ unless $\cup p = \cup q$.  As $\abs{\cup p} = 2r_1 \neq 2r_2 = \abs{\cup q}$, this cannot occur. This proves orthogonality.
\end{proof}

We conclude this section with calculating the norm of $M_r^{(m)}$. To do so, we first recall some useful relationships between spin operators. Let $\epsilon_{i,j,k}$ denote the Levi-Civita symbol, and define
\begin{equation} \label{eq:anti-symmetric}
    E_{a,b,c}= \sum_{i,j,k} \epsilon_{i,j,k}\, S^i_a \otimes S^j_b \otimes S^k_c,
\end{equation} 
which is antisymmetric under transpositions $\tau$ of the indices, i.e. $E_{\tau(a),\tau(b),\tau(c)}=-E_{a,b,c}.$ Using $S^i S^j = \frac{1}{4}\delta_{i,j} \idty + \frac{i}{2} \epsilon_{i,j,k} S^k$, the following relations hold:
  \begin{equation}\label{eq:heisenberg-contractions}
    (\heis{a}{b})^2 = \frac{3}{16}\idty - \frac{1}{2} \heis{a}{b},
    \quad (\heis{a}{b})(\heis{b}{c}) = \frac{1}{4} \heis{a}{c} - \frac{i}{2}E_{a,b,c}, \quad \forall a \neq b \neq c.
  \end{equation}

\begin{lemma}\label{lemma:matchings-norm}
  For all $0\le r\le \lfloor \frac{m}{2} \rfloor$, the Hilbert-Schmidt norm of the $r$-matching operator satisfies
  \begin{equation}
    \|M^{(m)}_r\|_2^2 = \qty(\frac{ (2r-1)!!}{2^{2r}})^2  (2r+1) \binom{m}{2r} 2^m.
  \end{equation}
\end{lemma}
\begin{proof}
Since the trace of $E_{a,b,c}$ over the $b$-th index is zero, it follows from \eqref{eq:heisenberg-contractions} that
  \begin{equation}\label{eq:cycle-contraction}
    \tr[ (\heis{a_1}{a_2})(\heis{a_2}{a_3})\cdots(\heis{a_k}{a_1})] = \frac{3}{4^k} \tr(\idty_{a_1,\dots,a_k}) = \frac{3}{2^k}.
  \end{equation}
As indicated in \eqref{eq:hs-product},
  \[
  \|M_r^{(m)}\|_2^2 = \sum_{\substack{p,q\in \mathcal{M}_r^{m}: \\ \cup p = \cup q}} \tr(S_p S_q).
  \]
where we again use that $\tr(S_p S_p) = 0$ unless $\cup p = \cup q$. Since $\tr(S_p S_q)$ is invariant under permutations of the factors, without loss of generality, assume that $\cup p = \cup q = \{1,\dots, 2r\}$. Thus, the norm calculation reduces to
  \begin{equation}\label{eq:M-norm-reduction}
      \|M^{(m)}_r\|_2^2 = \binom{m}{2r} \tr(\idty_{2r+1,\dots,m}) \|M^{(2r)}_r\|_2^2 = \binom{m}{2r} 2^{m-2r} \|M^{(2r)}_r\|_2^2.
  \end{equation}

To calculate $\|M^{(2r)}_r\|_2^2$, consider first an arbitrary pair $p,q\in\mathcal{M}^{2r}_r$, which we refer to as \emph{perfect matchings} over $2r$ elements. As we now show, pairs $(p,q)$ of perfect matchings of $2r$ elements are in bijection with permutations $\pi \in \mathfrak{S}_{2r}$ of $2r$ elements which have no odd length cycles, which we denote by $\mathfrak{S}_{2r}^e$. 

Given $(p,q)$, the corresponding permutation $\pi$ is built as follows: let $a_1\in [2r]$ be arbitrary, and inductively define $a_{2i}$ to be the element connected to $a_{2i-1}$ in $p$, and $a_{2i+1}$ to be the element connected to $a_{2i}$ in $q$. Then, by properties of the matching, it must be that there is some finite $l$ such that $a_{2l+1}=a_1$. Thus, $(a_1,\dots, a_{2l})$ is a cycle of length $2l$ which is added to to $\pi$. If there are no more elements, then $\pi$ is the desired permutation. Otherwise pick an element which has not yet been considered and iterate this process to create a new cycle. After a finite number of steps the iteration terminates with the desired permutation $\pi$.

Vice-versa, given $\pi \in \mathfrak{S}_{2r}^e$, define $(p,q)$ as follows: for each cycle $(a_1,\dots,a_{2l})$ in $\pi$, add $(a_{2i-1},a_{2i})$ to $p$ and $(a_{2i},a_{2i+1})$ to $q$, for each $i=0,\dots, l$ (where addition is taken modulo $2l$). Then $p$ and $q$ are perfect matchings as they both have $r$ disjoint pairs.

Let $\pi \in \mathfrak{S}_{2r}^e$ be the permutation corresponding to a pair $p,q \in \mathcal{M}_r^{2r}$, $N(\pi)$ the number of cycles of $\pi$, and and $\ell_1,\dots, \ell_{N(\pi)}$ be the length of the cycles. Then by \eqref{eq:cycle-contraction},
\[
\tr(S_p S_q) = \prod_{j=1}^{N(\pi)} \frac{3}{2^{\ell_j}} = \frac{3^{N(\pi)}}{2^{2r}},
\]
as $\ell_1+\dots+\ell_{N(\pi)} = 2r$. Thus, we have proved that
\begin{equation}
\|M_r^{(2r)}\|_2^2 = \frac{1}{2^{2r}} \sum_{\pi \in \mathfrak{S}_{2r}^e}  3^{N(\pi)}.
\end{equation}
Denoting by $h(2r,k)$ the number of permutations in $\mathfrak{S}_{2r}^e$ which have exactly $k$ cycles~\cite{A204420}, the last equation can be rewritten in terms of the generating function of $h(2r,k)$. This is computed in Lemma~\ref{lemma:gf-even-cyles}, which shows
\[
 \|M_r^{(2r)}\|_2^2 = \frac{1}{2^{2r}} \sum_{k=1}^{r} h(2r,k)3^{k} = \frac{1}{2^{2r}}  2^{r} \qty(\frac{3}{2})^{\bar r} (2r-1)!!, 
\]
where $\alpha^{\bar r}$ is the raising factorial. This evaluates to
\[
2^r \qty(\frac{3}{2})^{\bar r} =   2^r \frac{3}{2}\qty(\frac{3}{2}+1) \cdots \qty(\frac{3}{2} + r-1) = 3(3+2)\cdots (2r +1) = (2r+1)!!,
\]
which produces
\[  \|M_r^{(2r)}\|_2^2 = \qty(\frac{ (2r-1)!!}{2^r})^2 (2r+1). \]
Inserting this into \eqref{eq:M-norm-reduction} produces the final result.

\end{proof}

\subsection{A recursion relation for the matching operators}

The projection $P_{\rm sym}^{(m)}$ onto the highest-weight spin subspace naturally arises when considering the transfer operator associated with an undecorated AKLT model. The motivation for introducing the operators $M_r^{(m)}$ is that they form a basis for $\cZ^{(m)}$ and are eigenvectors of the transfer operator $\bE^{\otimes m}$. If one determines coefficients so that
\begin{equation}\label{eq:sym_decomp}
P_{\rm sym}^{(m)} = \sum_{r} c^{(m)}(r) M_r^{(m)},
\end{equation}
then (as we show in the next section) it will be possible to write down the SVD of transfer operator of the decorated AKLT model from the SVD of the transfer operator of the undecorated model. The goal of this section is to explicitly determine \eqref{eq:sym_decomp}. This result will be a consequence of the following recursion relation, which shows how to rewrite the product $M_1^{(m)} M_r^{(m)}$ as a sum of matching operators.

\begin{lemma}\label{lemma:product-m1-mr}
Fix $m\geq 2$, and for all $0 \le r \le s:=\floor{\half[m]}$ define
\begin{equation} \label{eq:recursion_coeffs}
	a_r = \frac{r}{2}(m-2r-1), \qquad b_r = r+1, \qquad c_r = \frac{2r+1}{16}\binom{m-2r+2}{2}.
\end{equation}
Then, the matching operators satisfy 
\begin{equation}\label{eq:Sr_recursion}
    M^{(m)}_1 M^{(m)}_r = c_r M^{(m)}_{r-1} 
    + a_r \, M^{(m)}_r 
    + b_r\, M^{(m)}_{r+1},
\end{equation}
where, for consistency, one takes $M_{-1}^{(m)}=M_{s+1}^{(m)}=0$.
\end{lemma}

We first prove the main technical result needed to prove Lemma~\ref{lemma:product-m1-mr}, which makes use of the partition elements $X_r(i,j)$ of $\cM_r^m$ from Definition~\ref{def:matching_partition}.

\begin{lemma}\label{lemma:summation} Fix $m\geq 2$. Then the following relations hold for any $0\leq r\leq \floor{\tfrac{m}{2}}$:
\begin{align}
  \sum_{i<j}  \sum_{p \in A_r(i,j)} S_p &= r \, M^{(m)}_r \label{eq:Asum}\\
  \sum_{i<j}  \sum_{p \in B_r(i,j)} S_p &= 2r(m-2r)\, M^{(m)}_r \label{eq:Bsum}\\
  \sum_{i<j}  \sum_{p \in C_r(i,j)} S_p &= 2r(r-1)  \, M^{(m)}_r \label{eq:Csum}\\
  \sum_{i<j}  \sum_{p \in D_r(i,j)} S_p &= \binom{m-2r}{2} \, M^{(m)}_r \label{eq:Dsum}
\end{align}
where summations over empty sets are by convention taken to be zero.
\end{lemma}

We note that for each of the sets from Definition~\ref{def:matching_partition}, there are one or two possible values of $r$ for which the set is empty. These are precisely the values of $r$ for which the coefficient on the RHS of \eqref{eq:Asum}-\eqref{eq:Dsum} is zero and so the equality holds.

\begin{proof}

By the previous remark, we need only consider the values of $r$ such that $X_r(i,j)$ is nonempty. This is independent of $(i,j).$

The identity \eqref{eq:Asum} follows from noting that any matching $p\in \cM_r^m$ belongs to precisely $r$ different sets $A_r(i,j)$, namely, those that are associated with the elements $(i,j)\in p$.

For \eqref{eq:Bsum}, every matching $p\in \cM_r^m$ belongs to precisely $2r(m-2r)$ different sets $B_r(i,j)$ labeled by taking a pair $(i,j)$ where one element belongs to $\cup p$ and one element belongs to $[m]\setminus \cup p$. The result follows.

To establish \eqref{eq:Csum}, begin by recalling that that $p\in C_r(i,j)$ if and only if $i,j\in \cup p$ but $(i,j)\notin p$. Hence, given any fixed $p\in \cM_r^m$ there are exactly $\frac{2r(2r-2)}{2} = 2r(r-1)$ distinct choices for the set $\{i,j\}$ such that $p\in C_r(i,j)$. These are obtained by first picking any $i\in\cup p$, then taking an element $j\in\cup p$ that belongs to an element $(j,k)\in p$ with $i\neq j,k$ and recalling that $(i,j)$ is unordered.

Finally, the equality in \eqref{eq:Dsum} is a consequence of the fact that $p\in D_r(i,j)$ if and only if $i,j\in [m]\setminus \cup p$. For any $r$-matching $p$, there are $\binom{m-2r}{2}$ distinct choices for $(i,j)$.

\end{proof}

\begin{proof}[Proof of Lemma~\ref{lemma:product-m1-mr}] We drop the superscript $m$ and set $S_{i,j} = \heis{i}{j}$ to simplify notation.

  For the case $r=1$, the result is a consequence of breaking up the summation \[(M_1)^2 = \sum_{i<j}\sum_{k<l}S_{i,j}S_{k,l}\] 
  in terms of $|\{i,j\}\cap\{k,l\}|$, and then applying the identities from \eqref{eq:heisenberg-contractions}. Since $E_{i,j,k}=-E_{k,j,i}$, this yields
  \begin{align*}
    (M_1)^2 & = \sum_{i<j}S_{i,j}^2 + \sum_{i<k}\sum_{j\neq i,k} \acomm{S_{i,j}}{S_{j,k}} +
    \sum_{i<j}\sum_{\substack{k<l : \\ k,l\neq i,j}} S_{i,j}S_{k,l} \\
    & = \frac{3}{16}\binom{m}{2} \idty + \qty(-\frac{1}{2} +
    \frac{m-2}{2})M_1 + 2 M_2 = \frac{3}{16}\binom{m}{2} \idty + \frac{m-3}{2} M_1 + 2 M_2 .
  \end{align*}
 
  Now, fix $r>1$. We write the product $M_1M_r$ as
  \begin{equation}
      \label{eq:M1Mr_general}
       M_1M_r = \sum_{i<j}\sum_{p\in\cM_r}S_{i,j}S_p = \sum_{X\in \{A, B, C, D\}}\sum_{i<j} \sum_{p\in X_r(i,j)} S_{i,j}S_{p}
  \end{equation}
and consider the cases $X\in \{A, B, C, D\}$ separately.
\medskip

Fix $i<j$ and consider the simplest case, $D_{r}(i,j)$. Since $D_{r}(i,j)\ni p \mapsto p\cup(i,j) \in A_{r+1}(i,j)$ is a bijection, it follows that
  \[
  \sum_{p \in D_r(i,j)} S_{i,j} S_p = \sum_{p \in A_{r+1}(i,j)} S_p.
  \]
  Summing over $i<j$ and applying Lemma~\ref{lemma:summation} yields
  \begin{equation} \label{eq:Dsum_final}
      \sum_{i<j} \sum_{p\in D_r(i,j)} S_{i,j}S_p = (r+1) M_{r+1}.
  \end{equation} 
\medskip

For the case $p\in A_r(i,j)$, one can write $S_p = S_{i,j} S_{p\setminus(i,j)}$ which by \eqref{eq:heisenberg-contractions} implies
    \[ S_{i,j} S_p = (S_{i,j})^2 S_{p\setminus(i,j)} = \frac{3}{16} S_{p\setminus(i,j)} - \frac{1}{2} S_{p} .\]
Thus, applying the bijection between $A_r(i,j)$ and $D_{r-1}(i,j)$ and then Lemma~\ref{lemma:summation} gives
    \begin{align}
        \sum_{i<j}\sum_{p \in A_r(i,j)} S_{i,j}S_p & = \frac{3}{16} \sum_{i<j}\sum_{p \in D_{r-1}(i,j)} S_p -\frac{1}{2} \sum_{i<j}\sum_{p \in A_r(i,j)} S_p \nonumber \\
        &= \frac{3}{16} \binom{m-2r+2}{2}  M_{r-1}  -\frac{1}{2} r M_r.\label{eq:ASum_final}
    \end{align}
\medskip

Now consider $B_r(i, j)$, for which one can write
 \[
 \sum_{p\in B_r(i,j)}S_{i,j}S_p = \sum_{x\neq i,j}\sum_{p'\in \cM_{r-1}^{m-3}}S_{i,j}\left(S_{i,x} + S_{j,x}\right)S_{p'}
 \]
 where $\cM_{r-1}^{m-3}$ is identified with the set of $r-1$-matchings on $[m]\setminus\{x,i,j\}$. Then, applying \eqref{eq:heisenberg-contractions} and using $E_{i,j,x}=-E_{j,i,x}$ one finds 
    \[
    S_{i,j}\left(S_{i,x} + S_{j,x}\right)S_{p'} = \frac{1}{4}\left(S_{j,x} + S_{i,x}\right)S_{p'}.
    \]
Summing the final expression over all $x$, $p'$ and then all $i<j$ produces the final identity,
 \begin{equation}\label{eq:BSum_final}
     \sum_{i<j} \sum_{p\in B_r(i,j)}S_{i,j}S_p = \frac{1}{2}r(m-2r)M_r
 \end{equation}
 where we again apply Lemma~\ref{lemma:summation}.

\medskip
For the case of $C_r(i,j)$, we begin by expanding
  \begin{equation}\label{eq:c_sum_ops}
  \sum_{p\in C_r(i,j)} S_{i,j}S_p = \sum_{\substack{x<y \, :\\ x,y\neq i,j}}\sum_{p'\in \cM_{r-2}^{m-4}}S_{i,j}(S_{i,x}S_{j,y}+S_{j,x}S_{i,y})S_{p'}
  \end{equation}
  where $\cM_{r-2}^{m-4}$ is identified with all $r-2$ pairings on $[m]\setminus\{i,j,x,y\}.$ Then, by \eqref{eq:heisenberg-contractions},
\begin{equation*}
S_{i,j}S_{i,x}S_{j,y}=\qty( \frac{1}{4} S_{j,x}  - \frac{i}{2}E_{j,i,x} )S_{j,y} = \frac{1}{16} S_{x,y} - \frac{i}{8} E_{x,j,y} -\frac{i}{2} E_{j,i,x} S_{j,y}.
\end{equation*}
Expanding $E_{j,i,x}$ and using $S^aS^b = \frac{1}{4}\delta_{a,b}\idty + \frac{i}{2}\epsilon_{a,b,c}S^c$, the last term above can be calculated as
\begin{align*}
-\frac{i}{2} E_{j,i,x} S_{j,y} & = -\frac{i}{8}E_{y,i,x} + \frac{1}{4}\sum_{b,c,d,e}\left(\sum_a \epsilon_{a,b,c}\epsilon_{a,d,e}\right) S_i^b\otimes S_j^e\otimes S_x^c\otimes S_y^d \nonumber\\
& =-\frac{i}{8}E_{y,i,x} + \frac{1}{4}\left(S_{i,y}S_{j,x}-S_{i,j}S_{x,y}\right)
\end{align*}
where we use that the Levi-Civita symbol satisifies $ \sum_a \epsilon_{a,b,c} \epsilon_{a,d,e} = \delta_{b,d}\delta_{c,e} - \delta_{b,e} \delta_{c,d}$. 

The analogous calculation holds for $S_{i,j}S_{j,x}S_{i,y}$ by exchanging $i$ and $j$ in the above formulas. Putting all of this together, and using again the anti-symmetric property of $E_{a,b,c}$, one finds
\begin{equation*}\label{eq:C_sum_term}
 S_{i,j}(S_{i,x}S_{j,y}+S_{j,x}S_{i,y})S_{p'} = \frac{1}{4}\left(S_{i,x}S_{j,y}+S_{j,x}S_{i,y}\right)S_{p'}+\left(\frac{1}{8}-\frac{1}{2}S_{i,j}\right)S_{x,y}S_{p'}
\end{equation*}
Comparing the first term above with \eqref{eq:c_sum_ops}, it is clear that summing over all possible $x,y$ and $p'$ produces 
\begin{equation}\label{eq:sum_part1}
\frac{1}{4}\sum_{\substack{x<y: \\x,y\neq i,j}}\sum_{p'\in \cM_{r-2}^{m-4}}\left(S_{i,x}S_{j,y}+S_{j,x}S_{i,y}\right)S_{p'} =\frac{1}{4}\sum_{p\in C_r(i,j)}S_p.
\end{equation}
While for the other term
\begin{equation}\label{eq:sum_part2}
  \sum_{\substack{x<y: \\x,y\neq i,j}}\sum_{p'\in \cM_{r-2}^{m-4}}\left(\frac{1}{8}-\frac{1}{2}S_{i,j}\right)S_{x,y}S_{p'} = \frac{r-1}{8}\sum_{p\in D_{r-1}(i,j)}S_p - \frac{r-1}{2}\sum_{p\in A_r(i,j)}S_p
\end{equation}
where we use $p'\cup (x,y)\in D_{r-1}(i,j)$, that each $p= \{(x_1,y_1), \ldots, (x_{r-1},y_{r-1})\}\in D_{r}(i,j)$ appears exactly $r-1$ times on the LHS, and the bijection \[A_{r}(i,j)\ni p\mapsto p\setminus(i,j)\in D_{r-1}(i,j).\]

Finally, summing \eqref{eq:sum_part1}-\eqref{eq:sum_part2} over $i<j$ and applying Lemma~\ref{lemma:summation} shows
\begin{align} \label{eq:CSum_final}
\sum_{i<j}\sum_{p\in C_{r}(i,j)} S_{i,j}S_p & = \frac{r-1}{8}\binom{m-2r+2}{2}M_{r-1}.
\end{align}
\medskip

Inserting \eqref{eq:Dsum_final}-\eqref{eq:BSum_final} and \eqref{eq:CSum_final} into \eqref{eq:M1Mr_general} produces the desired expression for $M_1M_r$.

\end{proof}

The recursion from Lemma~\ref{lemma:product-m1-mr} can be nicely rewritten in terms of a operator-valued vector product. Denoting by $s:=\floor{\half[m]}$ and following the convention for spin matrices, let $\bm{M}^{(m)}$ the vector of $s+1$ operators given by
\begin{equation}
\bm{M}^{(m)} = \mqty( M_0^{(m)}, M_1^{(m)}, \dots, M_s^{(m)} ).
\end{equation}
Then, for any $v \in \bC^{s+1}$, define
\begin{equation}
v \cdot \bm{M}^{(m)} = \sum_{r=0}^{s} v_r \, M^{(m)}_r \in \mathcal{Z}^{(m)}.
\end{equation}
Thus, $\mathcal{Z}^{(m)} = \{v\cdot \bm{M}^{(m)} : v\in \bC^{s+1}\}$ by Proposition~\ref{prop:matchings-are-basis}. With this notation, Lemma~\ref{lemma:product-m1-mr} can be rephrased as follows.
\begin{lemma}\label{lemma:product-m1} Fix $m\geq 3$, and let $B$ be the $(s+1)\times (s+1)$ tridiagonal matrix
    \begin{equation}\label{eq:B_matrix}
	B = \begin{pmatrix}
		a_0 & c_1 & 0 & \ldots & 0 \\
		b_0 & a_1 & c_2 & \ddots & 0 \\
		\vdots & \ddots & \ddots& \ddots & \vdots \\
		\vdots & \ddots & \ddots & \ddots & c_s \\
		0 & \ldots & \ldots & b_{s-1} & a_s
	\end{pmatrix}
\end{equation}
whose entries $a_r$, $b_r$, $c_r$ are defined in Lemma~\ref{lemma:product-m1-mr}. Then for each $v \in \bC^{s+1}$, 
\begin{equation}
    M_1^{(m)} (v \cdot \bm{M}^{(m)})= (B v) \cdot \bm{M}^{(m)}.
\end{equation}
\end{lemma}
\begin{proof}
By Lemma~\ref{lemma:product-m1-mr},
\begin{align*}
M_1^{(m)} (v \cdot \bm{M}^{(m)}) & =
\sum_{r=0}^{s} 
v_r\qty( c_r M_{r-1}^{(m)} + a_r M_r^{(m)} + b_r M_{r+1}^{(m)}) = \sum_{r=0}^{s} v'_r M_r^{(m)}
\end{align*}
where,
\[
    v'_r = \begin{cases} 
    c_1\, v_1 & \text{if $r = 0$,} \\ 
    b_{s-1} v_{s-1} + a_s v_{s} & \text{if $r = s$,} \\ 
    b_{r-1} v_{r-1} + a_r v_r + c_{r+1} v_{r+1} & \text{otherwise.}
    \end{cases}
\]
Since $a_0=0$, this is equivalent to $v' = Bv$.
\end{proof}

The following are two immediate consequences of Lemma~\ref{lemma:product-m1}:

\begin{cor}\label{cor:B_properties} Fix $m\geq 3$ and define $e_0 = (1,0,\ldots, 0)\in \bC^{s+1}$. The matrix $B$ from Lemma~\ref{lemma:product-m1} satisfies the following two properties:
\begin{enumerate}
    \item $q(M_1^{(m)}) = (q(B)e_0)\cdot \bm{M}^{(m)}$ for any polynomial $q\in \bC[t]$,
\item The set $\{e_0, \ldots, e_s\}$ is a basis for $\bC^{s+1}$ where $e_r = B^re_0.$
\end{enumerate}
\end{cor}

\begin{proof}
Since $e_0 \cdot \bm{M}^{(m)} = \idty$, it follows from Lemma~\ref{lemma:product-m1} that
\[ 
    (M_1^{(m)})^p = (B^p e_0) \cdot \bm{M}^{(m)}, \qquad p \geq 0,
\]
which extends to arbitrary polynomials $q$ by linearity. Moreover, for any $0\leq r \leq s$, by the definition of $B$ in Lemma~\ref{lemma:product-m1}, $e_r = (e_r(0), \ldots, e_r(s))$ where
\[		e_r(r) = r!\;\;\text{and}\;\; e_r(p)=0 \;\;\text{for}\;\; p>r.
\]
\end{proof}

From this corollary, an approach for calculating \eqref{eq:sym_decomp} starts to emerge. From \eqref{eq:Casimir_oneMatching}, it is trivial that
\[
\spec(M_1^{(m)}) = \left\{\lambda_j := \frac{j(j+1)}{2}-\frac{3m}{8} : j \in J_m\right\},
\]
As a consequence, for any polynomial $p_j$ such that $p_j(\lambda_j)= 1$ and $p_j(\lambda_k)=0$ for all $k \neq j$, one has
\begin{equation}\label{eq:strategy-q}
  Q^{(m,j)} = p_j(M_1^{(m)}) = (p_j(B)e_0) \cdot \bm{M}^{(m)}.
\end{equation}
Such a polynomial $p_j$ is easy to find, e.g. $p_j(t) = \frac{\prod_{i\neq j} t-\lambda_i}{\prod_{i\neq j}\lambda_j-\lambda_i}.$ However, now one needs to evaluate $p_j(B)$, and so we turn to considering the spectrum and eigenvectors of $B$.

\begin{lemma}\label{lemma:spectrum-b}
The spectrum of $B$ and of $M_1^{(m)}$ are the same up to multiplicities.
If $w \in (\bC^2)^{\otimes m}$ is an eigenvector of $M_1^{(m)}$, then 
$y \in \bC^{s+1}$ defined by
\begin{equation}\label{eq:B_eigvec}
y(r) = \bra{w} M_r^{(m)} \ket{w} \qc r = 0,\dots, s:= \floor{\tfrac{m}{2}}
\end{equation}
is a left eigenvector of $B$ with the same eigenvalue.
\end{lemma}
\begin{proof}
Note that $\spec(M_1^{(m)})$ has $s+1$ distinct real eigenvalues. Hence, one only needs to show that each eigenvalue of $M_1^{(m)}$ is an eigenvalue of $B$ with the corresponding left eigenvector.

Let $\lambda \in \spec(M_1^{(m)})$ and $w\in (\bC^2)^{\otimes m}$ be a corresponding eigenvector. Then, for all $0\leq p \leq s$,
\[
\braket{w}{q_p(M_1^{(m)})w} = 0, \quad \text{where} \quad q_p(x) = (x-\lambda)x^p.
\]
Defining $v_p := (B-\lambda)e_p$, it follows from Corollary~\ref{cor:B_properties} that
\[ q_p(M_1^{(m)}) = (q_p(B) e_0) \cdot \bm{M}^{(m)} = v_p \cdot \bm{M}^{(m)}.\]
Since $q_p(M_1^{(m)})w=0$ for all $p$, and the vectors $v_p$ and $y$ have real entries, the above equality implies
\begin{equation}\label{eq:zero_ip}
  0= \braket{w}{q_p(M_1^{(m)})w} =  \sum_{r=0}^s v_p(r) \braket{w}{M_r^{(m)}w} =  \sum_{r=0}^s v_p(r) y(r) = \braket{y}{v_p},  
\end{equation}
where we use the standard inner product on $\bC^{s+1}.$ The vector $y$ is nonzero since $y(0) = \|w\|^2 \neq 0$. Hence, \eqref{eq:zero_ip} is equivalent to
\[0 = \braket{y}{v_p} = \braket{(B^t-\lambda)y}{e_p} \quad \forall\, 0 \leq p \leq s.\]
As $\{e_0, \ldots, e_s\}$ is a basis for $\bC^{s+1}$, it must be that $(B^t - \lambda)y= 0$, i.e., $y$ is a left eigenvector of $B$ with eigenvalue $\lambda$.
\end{proof}

While Lemma~\ref{lemma:spectrum-b} characterizes the left eigenvectors of $B$, the next result explains how to obtain the corresponding right eigenvectors.
\begin{lemma}\label{lemma:eigenvectors-b}
Fix $m\geq 2$ and let $D=\operatorname{diag}(d_0, \dots, d_{s})\in M_{s+1}(\bC)$ be the diagonal matrix whose entries are given by
    \begin{equation}
        d_0 = 1 \qc d_r = \frac{b_{r-1}}{c_r} d_{r-1} \quad r = 1,\dots, s.
    \end{equation}
    If $y_r$ is a left eigenvector of $B$ with eigenvalue $\lambda_r$, then $Dy_r$ is a right eigenvector of $B$ with the same eigenvalue. As a consequence,
    \begin{equation}\label{eq:spectral-decomposition-b}
        B = \sum_{r=0}^s \lambda_r \frac{1}{\expval{D}{y_r}} D \dyad{y_r},
    \end{equation}
    and $\mel{y_{r}}{D}{y_{r'}} = \delta_{r,r'}\expval{D}{y_r}$ for every $r,r' = 0,\dots,s$.
\end{lemma}

We note that $D$ is a strictly positive matrix since $b_{r-1},c_r>0$ for all $1\leq r\leq s$. As such, the square root and inverse of $D$ are well-defined.

\begin{proof}
Let $y$ be a left eigenvector of $B$ with eigenvalue $\lambda$. By definition, it satisfies
\begin{equation}\label{eq:Bt_eigenvector}
   (B^t y)(r) = c_r y(r-1) + a_r y(r) + b_r y(r+1) = \lambda y(r) \quad \forall 0 \le r \le s,
\end{equation}
where the cases $r=0$ and $r=s$ are interpreted by setting $y(-1) = y(r+1) = 0$. Now consider $y' = Dy$. Then,
\begin{align*}
    (By')(r) & = b_{r-1} y'(r-1) + a_r y'(r) + c_{r+1} y'(r+1) \\
             & = b_{r-1} d_{r-1} y(r-1) + a_rd_r y(r) + c_{r+1} d_{r+1}y(r+1).
\end{align*}
The definition of the diagonal elements are such that $d_{r-1}/d_r  = c_r/ b_{r-1}$ for $1\leq r \leq s$ which, considering \eqref{eq:Bt_eigenvector}, implies $y'=Dy$ is a right eigenvector of $B$ as then
\begin{equation}\label{eq:B_right_eigenvector}
 (By')(r) = d_r( c_r y(r-1) + a_r y(r) + b_r y(r+1)) = \lambda y'(r).
\end{equation}

Denote by $\spec(B)=\{\lambda_r : 0 \leq r \leq s\}$ and let $\{y_r : 0 \leq r \leq s\}$ be the corresponding set of left eigenvectors. It is easy to verify that $B'=D^{-1/2}BD^{1/2}$ is symmetric, and hence has an orthonormal basis of eigenvectors. Since each eigenvalue of $B$ is simple by Lemma~\ref{lemma:spectrum-b}, \eqref{eq:B_right_eigenvector} implies that
\[
v_r = \frac{1}{\norm{ D^{1/2} y_r}} D^{1/2} y_r, \qquad 0 \leq r \leq s
\]
is an orthonormal eigenbasis of $B'$. Hence, $\mel{y_{r}}{D}{y_{r'}} = \delta_{r,r'}\expval{D}{y_r}$, and by the spectral theorem
\[
B = D^{1/2}B'D^{-1/2} = \sum_{r=0}^s \lambda_r \frac{1}{\expval{D}{y_r}} D \dyad{y_r}.
\]

\end{proof}

This final decomposition of $B$ allows us to calculate $p(B)$ for any polynomial $p$. As such, we can determine the values of the coefficients from \eqref{eq:sym_decomp}.

\begin{thm}\label{thm:Psym_Sr}
 	For any $m\geq 2$, the projection onto the symmetric subspace $P^{(m)}_{\rm sym}$ is
 	\begin{equation}\label{eq:Psym_matchings_sum}
 		P^{(m)}_{\rm sym} = \sum_{r=0}^{\lfloor\frac{m}{2} \rfloor}c^{(m)}(r) M_r^{(m)}, \qquad c^{(m)}(r) = \frac{m+1}{2^m}\frac{4^{r}}{(2r+1)!!}.
 	\end{equation}
 \end{thm}

\begin{proof} Let $s= \floor{\tfrac{m}{2}}$ and $\lambda_s = \max \spec(M_1^{(m)})$. Since $P_{\rm sym}^{(m)}= Q^{(m,m/2)}$, it follows from \eqref{eq:strategy-q} that

    \[ 
        P^{(m)}_{\rm sym} = p(M_1^{(m)}) = (p(B)e_0) \cdot \bm{M}^{(m)},
    \]
    where $p$ is such that $p(\lambda_s)=1$ and $p(\lambda) = 0$ for all $\lambda\in \spec(M_1^{(m)})\setminus\{\lambda_s\}$. Thus, $ c^{(m)}(r) = (p(B)e_0)(r)$, and so by Lemma~\ref{lemma:eigenvectors-b}, 
    \[
    p(B)e_0 = \frac{1}{\expval{D}{y_s}} D \ket{y_s} \braket{y_s}{e_0} = \frac{y_s(0)}{\expval{D}{y_s}} D \ket{y_s}.
    \]
The proof is completed by determining $y_s$, $y_s':=Dy_s$, and their scalar product.

Since $w = \ket{\uparrow\ldots \uparrow}$ is an eigenvector of $M_1^{(m)}$ associated with $\lambda_s$, Lemma~\ref{lemma:spectrum-b} implies that the corresponding left eigenvector of $B$ is defined by
\begin{equation}\label{eq:right-largest-eigenvector}
y_s(r) = \bra{\uparrow\ldots \uparrow} M_r^{(m)} \ket{\uparrow\ldots \uparrow}=  \binom{m}{2r} \frac{(2r-1)!!}{4^r}.
\end{equation}
where we have used \eqref{eq:M-nonzero} and \eqref{eq:matching_count}. In particular, $y_s(0) = 1$.

Now, by  the definition of $D,$ $y'_s(0) = 1$ and $y'_s(r) = d_r y_s(r)$ for $1\leq r \leq s$ where
\[
d_r = \prod_{i=1}^{r} \frac{b_{i-1}}{c_i} =  \frac{4^{2r} r!}{(2r+1)!!} \prod_{i=1}^r  \frac{2}{(m-2i+2)(m-2i+1)}.
\]
This can be further simplified by recognizing $\prod_{i=1}^r  (m-2i+2)(m-2i+1) = \frac{ m! }{(m-2r)!},$ which yields
\[
d_r =  \frac{4^{2r}}{(2r+1)!!} \cdot \frac{2^r r!(m-2r)!}{m!} = \frac{4^{2r}}{(2r+1)!! (2r-1)!! \binom{m}{2r}}.
\]
Thus, $y'_s(r) = \frac{4^r}{(2r+1)!!}$ and
\begin{equation}\label{eq:scalar-product-largest-eigenvectors}
\braket{y_s}{y'_s} = \expval{D}{y_s} =  \sum_{r=0}^s \frac{1}{2r+1} \binom{m}{2r} = \frac{2^m}{m+1}.
\end{equation}
The last equality can be seen, e.g., by first using $\frac{1}{2r+1}\binom{m}{2r}=\frac{1}{m+1}\binom{m+1}{2r+1}$ and then applying $\binom{m+1}{2r+1}=\binom{m}{2r+1}+\binom{m}{2r}$. Combining these produces the desired result:
\[
\left(p(B)e_0\right)(r) =  \frac{y'_s(r) y_s(0) }{\expval{D}{y_s}} =  \frac{m+1}{2^m} \frac{4^r}{(2r+1)!!}.
\]
\end{proof}

\section{The uniform spectral gap}\label{sec:Gap_Proof}

The proof of Theorem~\ref{thm:explicit_bound} follows from Proposition~\ref{prop:espilon-estimate} together with a sufficiently tight upper bound on $\|\E_L^{(d,n)}\|_\infty$ and lower bound on $q_L^{(d,n)}$. Recalling that $Q_L = \bE_L^{(d,n)}(\idty)$, both of these bounds will be a consequence of writing
\begin{equation}\label{eq:n-to-0_relation}
	\bE_L^{(d,n)} = (\bE^{\otimes d-1})^{n} \circ \bE^{(d,0)}
\end{equation}
in terms of its singular value decomposition. The goal of Section~\ref{sec:SVD} is to determine this and obtain the desired bounds. The proof of Theorem~\ref{thm:explicit_bound} is given in Section~\ref{sec:TheActualProof}.

\subsection{Bounding $\|\bE_L^{(d,n)}\|_\infty$ and $q_L^{(d,n)}$} \label{sec:SVD}

Recall that $\rho=\idty/2$. Since $\bE_L^{(d,n)}: M_2(\bC) \to M_2(\bC)^{\otimes d-1}$ and $\cB := \{\idty, \, \sigma^X, \, \sigma^Y, \, \sigma^Z\}$ is an orthogonal basis with respect to the Hilbert-Schmidt inner product, one can write
\begin{equation}\label{eq:EL_SVD_decomp}
\E_L^{(d,n)} = \ketbra*{\E_L^{(d,n)}(\idty)}{\rho} + \sum_{U=X,Y,Z} \ketbra*{\E_L^{(d,n)}(\sigma^U)}{S^U}.
\end{equation}

What is special in this case (and is proved in Lemma~\ref{lem:SVD_Decomp} below) is that
\[\left\{\E_L^{(d,n)}(\idty),\,\E_L^{(d,n)}(\sigma^X), \,\E_L^{(d,n)}(\sigma^Y),\, \E_L^{(d,n)}(\sigma^Z) \right\}\subseteq M_2(\bC)^{\otimes d-1}\] 
is also an orthogonal set. Therefore, up to constants, \eqref{eq:EL_SVD_decomp} is the desired SVD. This result is stated with respect to the following matrices, which depend on a parameter $\alpha \in \bC$:
\begin{align}
	V^{(d-1)}_{\rho}(\alpha) &=  \sum_{r=0}^{\floor{\half[d-1]}} \frac{\alpha^{2r}}{(2r+1)!!} M_r^{(d-1)} \label{eq:V1}\\
	V^{(d-1)}_{S^U}(\alpha)	&= -\sum_{r=0}^{\floor{\half[d-2]}} \sum_{j=1}^{d-1} \frac{\alpha^{2r+1}}{(2r+3)!!} S^U_j\otimes (M_{r}^{(d-2)})_{[d-1]\setminus \{j\}} \label{eq:VU}
\end{align}
Above, $U \in \{X,Y,Z\}$ and the notation $(M_{r}^{(d-2)})_{[d-1]\setminus \{j\}} $ denotes the matching operator associated with the $d-2$ indices that remain after removing $j$ from $[d-1]$.

\begin{lemma}[SVD of $\bE_L^{(d,n)}$]\label{lem:SVD_Decomp} Let $d\geq 3$, $n\geq 0,$ and set $\alpha_n = 2(-1/3)^n$. Then, 
\begin{equation}\label{eq:EL_on_basis}
	\bE_L^{(d,n)}(A) = \frac{d+1}{2^{d-1}}V_{A/2}^{(d-1)}(\alpha_n), \qquad \forall \, A \in \cB.
\end{equation}
The set $\{V_{A/2}^{(d-1)}(\alpha_n) : A\in \cB \}$ is orthogonal, and the singular values of $\bE_L^{(d,n)}$ are
\begin{equation}
	\left\{ \frac{d+1}{2^{d-1/2}} \| V_{A/2}^{(d-1)}(\alpha_n)\|_{2} : A\in \cB \right\}.
\end{equation}
\end{lemma}

\begin{proof}
Notice that it is sufficient to prove the result for $n=0$, as $\bE(S^U) = -\frac{1}{3} S^U$ for each $U=X,Y,Z$ implies
	\begin{align*}
  \bE^{\otimes d-1} ( M_r^{(d-1)} ) = \frac{1}{3^{2r}} M_r^{(d-1)}, \quad  \bE^{\otimes d-1} ( S^U\otimes M_r^{(d-2)}) = -\frac{1}{3^{2r+1}} S^U\otimes M_r^{(d-2)},
	\end{align*}
from which the general case follows from \eqref{eq:n-to-0_relation}. To simplify notation, we drop the superscript and write $\bE_L$ for $\bE_L^{(d,0)}$ in the remainder of this proof. 

Considering \eqref{eq:EL_SVD_decomp}, we proceed by calculating the Choi matrix $\tau$ corresponding to $\bE_L$ as a simple, but tedious, calculation shows
\begin{equation}\label{eq:Choi}
\tau := \sum_{i,j \in \{ \uparrow, \downarrow\}}  \bE_L(\dyad{i}{j}) \otimes \dyad{i}{j} = \sum_{A\in \cB}  \bE_L(A) \otimes \overline{A}/2.
\end{equation}
Using the  the definition of $\bE_L$ \eqref{eq:undecorated_transfer_op} and the recursion relation~\eqref{eq:symmetric-vector-recursion}, one quickly finds
\begin{align}
\sum_{i,j \in \{ \uparrow, \downarrow\}}  \bE_L(\dyad{i}{j}) \otimes \dyad{i}{j} & = (\idty_{2^{d-1}} \otimes K) P^{(d)}_{\text{sym}} (\idty_{2^{d-1}} \otimes K^*)\nonumber\\
& =  \sum_{r=0}^{\lfloor\frac{d}{2} \rfloor}c^{(d)}(r) (\idty_{2^{d-1}} \otimes K) M_r^{(d)} (\idty_{2^{d-1}} \otimes K^*) \label{eq:choi_2}
\end{align}
where the last equality follows from Theorem~\ref{thm:Psym_Sr}. We use properties of the set of matchings to rewrite $\tau$ in the latter form from \eqref{eq:Choi}.

Fix $0 \leq r \leq \floor{d/2}$. For each matching $p\in \mathcal M^d_r$ either $(j,d)\in p$ for some $1\leq j \leq d-1$, or $d\notin \cup p$ (in which case $p\in \cM_r^{d-1}$). Therefore, $\cM_r^{d}$ can be partitioned as
\begin{equation}
	\mathcal M^d_r = \mathcal M^{d-1}_r \cup \bigcup_{j=1}^{d-1} \left\{ \{(j,d)\} \cup x \mid x \in \mathcal M^{d-2}_{r-1} \right\}
\end{equation}
where $\cM^{d-2}_{r-1}$ is interpreted as the set of matchings on $[d-1]\setminus \{j\}$. The corresponding matching operator then factorizes as
\begin{equation}\label{eq:Mr_last_factor}
	M_r^{(d)} = M_r^{(d-1)} \otimes \idty + \sum_{j=1}^{d-1} \pi_{(j,d-1)}[ M_{r-1}^{(d-2)}\otimes (\heis{d-1}{d}) ]
\end{equation}
where $\pi_{(j,d-1)}$ is the permutation which swaps the tensor factors $d-1$ and $j$, and any previously undefined matching operators are taken to be zero, e.g. $M_{-1}^{(d-2)}=0$. 

Substituting \eqref{eq:Mr_last_factor} into \eqref{eq:choi_2}, one then finds
\begin{align}
\tau & = \sum_{r=0}^{\lfloor \frac{d}{2} \rfloor} c^{(d)}(r)\qty[ M_r^{(d-1)} \otimes \idty + \sum_{j=1}^{d-1} \pi_{(j,d-1)}[ M_{r-1}^{(d-2)} \otimes (\bm{S}_{d-1} \cdot K\bm{S}_d K^*)] ]. \label{eq:matchings_tau}
\end{align}
Consider the two summations above separately. Inserting the value of $c^{(d)}(r)$ from Theorem~\ref{thm:Psym_Sr}, the first summation can be rewritten as:
\begin{align}
\sum_{r=0}^{\lfloor \frac{d}{2} \rfloor} c^{(d)}(r)M_r^{(d-1)} \otimes \idty  
& = \frac{d+1}{2^d}\sum_{r=0}^{\lfloor \frac{d-1}{2} \rfloor} \frac{2^{2r}}{(2r+1)!!} M_r^{(d-1)} \otimes \idty \nonumber \\
& = \frac{d+1}{2^{d-1}}V_{\rho}^{(d-1)}(2)\otimes \overline{\idty/2}. \label{eq:SVD1}
\end{align}
The bound change in the first equality follows since either $\floor{\frac{d}{2}} = \floor{\frac{d-1}{2}}$, or $\floor{\frac{d}{2}}>\frac{d-1}{2}$ in which case the matching operator $M^{(d-1)}_{\floor{\frac{d}{2}}}=0.$

For the remaining terms from \eqref{eq:matchings_tau}, applying the relation $KS^UK^* = -\overline{S^U}$ and again inserting $c^{(d)}(r)$ shows
\begin{align}
\sum_{r=0}^{\lfloor \frac{d}{2} \rfloor} c^{(d)}(r)\sum_{j=1}^{d-1} &\pi_{(j,d-1)}[ M_{r-1}^{(d-2)} \otimes (\bm{S}_{d-1} \cdot K\bm{S}_d K^*)] \nonumber\\
& = -\frac{d+1}{2^{d-1}}\sum_{r=1}^{\lfloor \frac{d-2}{2} \rfloor +1} \frac{2^{2r-1}}{(2r+1)!!} \sum_{j=1}^{d-1} \pi_{(j,d-1)}[ M_{r-1}^{(d-2)} \otimes (\bm{S}_{d-1} \cdot \overline{\bm{S}_d})] \nonumber \\
& = \frac{d+1}{2^{d-1}}\sum_{U=X,Y,Z} V_{S^U}^{(d-1)}(2)\otimes \overline{S^U}, \label{eq:SVD2}
\end{align}
where, for the last line we have used that
\[
 \pi_{(j,d-1)}[ M_{r-1}^{(d-2)} \otimes (\bm{S}_{d-1} \cdot \overline{\bm{S}_d})]  = \left( M_{r-1}^{(d-2)} \right)_{[d-1]\setminus\{j\}}\otimes (\bm{S}_{j} \cdot \overline{\bm{S}_d}).
\]
Inserting \eqref{eq:SVD1}-\eqref{eq:SVD2} into \eqref{eq:matchings_tau} and comparing with \eqref{eq:Choi} establishes \eqref{eq:EL_on_basis}. 

To see that the set of operators $\{V_{A/2}^{(d-1)}(2): A\in \cB\}$ forms an orthogonal set, fix $p=\{(i_1,j_1), \ldots, (i_r,j_r)\}\in \cM_r^m$ and recall that by \eqref{eq:matching_ops}
\[
S_p = \sum_{U_1, \ldots, U_r}\bigotimes_{k=1}^r S_{i_k}^{U_k}\otimes S_{j_k}^{U_k},
\]
is a sum of simple tensors, each of which has an even number of the spin operator $S^U$ for all $U= X,Y,Z$. As the $r$-th matching operator is the sum of all such $S_p$, by  \eqref{eq:V1} $V_\rho^{(d-1)}(2)$ is also a sum of simple tensors, each of which has an even number of the spin operator $S^U$ for all $U$. The same argument shows that $V_{S^U}^{(d-1)}(2)$ is a sum of simple tensors with an odd number of $S^U$, and an even number of $S^V$ for $V\neq U$. Since $\cB$ is an orthogonal basis of $M_2(\bC)$, the orthogonality claim is a consequence of these observations. The set of singular values then follows from normalizing \eqref{eq:EL_SVD_decomp} appropriately.

\end{proof}

We now produce the necessary bounds on $q_L^{(d,n)}$ and $\|\bE_L^{(d,n)}\|_\infty$ to prove Theorem~\ref{thm:explicit_bound}.

\begin{lemma}\label{lem:quantity_bounds} Fix $d\geq 3$. The minimal eigenvalue $q_L^{(d,n)}$ of $Q_L^{(d,n)}=\bE_L^{(d,n)}(\idty)$ satisfies
	\begin{equation}\label{eq:q_L}
			q_L^{(d,n)} \geq  \frac{d+1}{3\cdot 2^{d-1}}\left(4- (1+3^{-n+1/2})^{d-1}\right).
	\end{equation}
In particular, $Q_L^{(d,n)}$ is invertible for any $n\geq \frac{\ln(d-1)}{\ln(3)}-\frac{\ln(\ln(3))}{\ln(3)}+\frac{1}{2}$. For such $n$, one also has
	\begin{equation}\label{eq:EL_infty_norm}
		\|\bE_L^{(d,n)}\|_{\infty}^2 = \frac{(d+1)^2}{2^{2d-1}}\|V_{\rho}^{(d-1)}(\alpha_n)\|_2^2 \leq \frac{(d+1)^2}{3\cdot 2^d}\left(2+(1+3^{-2n})^{d-1}\right).	
	\end{equation}
\end{lemma}

By Lemma~\ref{lem:SVD_Decomp}, it is simple to calculate the spectrum of $Q_L^{(d,n)}$ and $\|\E_L^{(d,n)}\|_\infty$ directly for small $d$ using Lemma~\ref{lemma:matchings-norm} and \eqref{eq:V_idty_norm}. Here, is also convenient to use \eqref{eq:Casimir} and \eqref{eq:Casimir_oneMatching}. For $1\leq d \leq 4$, this produces the values in Table~\ref{tab:small_d}. A similar calculation can be done for other small values of $d$ by apply the recursion relation from Lemma~\ref{lemma:product-m1} to write $M_r^{(d-1)}$ as a polynomial in $M_1^{(d-1)}$ and again invoking the relationship to the Casimir operator. 

Furthermore, using the recursion relation one can show that the eigenvalue of $M_r^{(d-1)}$ corresponding to the subspace of minimal spin $j_0$ is 
\begin{equation}
    (-1)^r \binom{ \floor{\half[d-1]}}{r} \frac{(2r+1)!!}{4^r}, 
\end{equation}
from which \eqref{eq:V1} implies that
\begin{equation}
    \frac{d+1}{2^{d-1}} \qty(1-3^{-2n})^{\floor{\half[d-1]}} \in \spec Q_L^{(d,n)}.
\end{equation}
We conjecture that this is $q_L^{(d,n)}$, which has been verified with the help of a computer algebra system for $d \le 30$. This would imply that $Q_L^{(d,n)}$ is invertible when $n\ge 1$ for all values of $d$. It also implies that the function $f(d)$ in Theorem~\ref{thm:explicit_bound} can be improved, although doing so would not change the asymptotic scaling of $n(\Delta(G))$.

\begin{table}
	\begin{center}
		\begin{tabular}{|c|c|c|}
			\hline
			$d$ & $q_L^{(d,n)}$ &  $\|\E_L^{(d,n)}\|_\infty$ \\
			\hline
            1 & $2$ & $\sqrt{2}$ \\
            2 & $\frac{3}{2}$ & $\frac{3}{2}$ \\
			3 & $1-3^{-2n}$ &  $\sqrt{2}(1+3^{-4n-1})^{1/2}$\\
			4 & $\frac{5}{8}(1-3^{-2n})$ & $\frac{5}{4}(1+3^{-4n})^{1/2}$\\
			\hline
		\end{tabular}
	\end{center}
	\caption{Exact values of $q_L^{(d,n)}$ and $\|\bE_L^{(d,n)}\|_\infty$ for $1\leq d \leq 4$.}
	\label{tab:small_d}
\end{table}

\begin{proof}
	Note that Lemma~\ref{lem:SVD_Decomp} and \eqref{eq:EL_SVD_decomp} imply $Q_L^{(d,n)} = \frac{d+1}{2^{d-1}}V^{(d-1)}_{\rho}(\alpha_n)$, and define
	\[
	R_L := V^{(d-1)}_{\rho}(\alpha_n) - \idty = \sum_{r=1}^{\floor{\frac{d-1}{2}}} \frac{\alpha_n^{2r}}{(2r+1)!!} M_r^{(d-1)}.
	\]
Hence, 
\begin{equation}\label{eq:spectral_bounds}
\spec Q_L^{(d,n)}\subseteq \left[\frac{d+1}{2^{d-1}}(1-\|R_L\|), \frac{d+1}{2^{d-1}}(1+\|R_L\|)\right].
\end{equation}
To bound $\|R_L\|$, first use $\|\bm{S}\cdot\bm{S}\| = \frac{3}{4}$ and \eqref{eq:matching_count} to bound the operator norm
\begin{equation} \label{eq:M_op_bound}
 \|M_r^{(m)}\| \leq \frac{3^r}{4^{r}}(2r-1)!!\binom{m}{2r}.   
\end{equation}
Then, since $\alpha_n^2 = 4\cdot 3^{-2n}$, the operator norm of $R_L$ is bounded by
\begin{align}
	\|R_L\| & \leq \sum_{r=1}^{\floor{\frac{d-1}{2}}} \frac{\alpha_n^{2r}}{(2r+1)!!} \|M_r^{(d-1)}\|  \leq \frac{1}{3}\sum_{r=1}^{d-1}\frac{1}{3^{(n-1/2)r}} \binom{d-1}{r} = \frac{(1+\frac{\sqrt{3}}{3^n})^{d-1}-1}{3} \label{eq:RL_bound}
\end{align}

Hence, \eqref{eq:q_L} holds from substituting \eqref{eq:RL_bound} into the lower bound from \eqref{eq:spectral_bounds}. That $Q_L^{(d,n)}$ is invertible follows from using $\binom{d-1}{r} \leq \frac{(d-1)^r}{r!}$ to further bound \eqref{eq:RL_bound} by
\[ 
\frac{1}{3}\sum_{r=1}^{d-1}\frac{1}{3^{(n-1/2)r}} \binom{d-1}{r} \leq  \frac{e^{\frac{d-1}{3^{n-1/2}}}}{3}
\]
which is less than one when $n > \tfrac{\ln(d-1)}{\ln(3)}-\tfrac{\ln(\ln(3))}{\ln(3)}+\frac{1}{2}.$

For the largest singular value of $\bE_L^{(d,n)}$, Lemma~\ref{lem:SVD_Decomp} implies that for any choice of $U$:
\begin{equation}\label{eq:EL_norm_general}
\|\bE_L^{(d,n)}\|_{\infty}^2 = \frac{(d+1)^2}{2^{2d-1}}\max\left\{\|V_{\rho}^{(d-1)}(\alpha_n)\|_2^2, \, \|V_{S^U}^{(d-1)}(\alpha_n)\|_2^2\right\}.
\end{equation}
In the case of $V_{\rho}^{(d-1)}(\alpha_n)$, applying the mutual orthogonality of the matching operators and Lemma~\ref{lemma:matchings-norm} it is straightforward to calculate
\begin{align}
	\|V^{(d-1)}_{\rho}(\alpha_n)\|_2^2 & = 
	2^{d-1}\sum_{r=0}^{\floor{\half[d-1]}} \qty(\frac{\alpha_n}{2})^{4r} \frac{1}{2r+1}\binom{d-1}{2r}  \nonumber \\
	& = \frac{2^{d-1}}{d}\sum_{r=0}^{\floor{\half[d-1]}} 3^{-4rn} \binom{d}{2r+1}  \label{eq:V_idty_norm}.
\end{align}
where in the last equality we have substituted $\alpha_n = 2\cdot(-3)^{-n}$.

In the case of $V_{S^U}^{(d-1)}(\alpha_n)$, first define
\begin{align*}
	W^{(d-1)}_{S^U}(r)  & := \sum_{j=1}^{d-1} S^U_j\otimes (M_{r}^{(d-2)})_{[d-1]\setminus \{j\}}, \quad 0 \leq r \leq \floor{\tfrac{d-2}{2}}
\end{align*}
As the matching operators are mutually orthogonal, these are also mutually orthogonal with respect to the Hilbert-Schmidt norm. Therefore,
\[
\|V_{S^U}^{(d-1)}(\alpha_n)\|_2^2 = \sum_{r=0}^{\floor{\frac{d-2}{2}}}\left(\frac{\alpha_n^{2r+1}}{(2r+3)!!}\right)^2\|W^{(d-1)}_{S^U}(r)\|_2^2.
\]
The permutation invariance of $W^{(d-1)}_{S^U}(r)$ moreover implies
\begin{align*}
\|W_{S^U}^{(d-1)}(r)\|_2^2  = & \frac{d-1}{2}\|M_r^{(d-2)}\|_2^2 + \binom{d-1}{2}\left\langle  S^U_{d-1}\otimes M_{r}^{(d-2)},  S^U_{d-2}\otimes (M_{r}^{(d-2)})_{[d-1]\setminus \{d-2\}} \right\rangle .
\end{align*}
The remaining inner product can be calculated by once again using \eqref{eq:Mr_last_factor} to decompose the matching operators over an appropriate index, and then invoking the orthogonality of the basis $\cB$. This produces the recursive formula
\[
\|W_{S^U}^{(d-1)}(r)\|_2^2 = \frac{d-1}{2}\|M_r^{(d-2)}\|_2^2 +\frac{(d-1)!}{2^2(d-3)!}\|W_{S^U}^{(d-3)}(r-1)\|_2^2.
\]
Iterating this identity $r$-times and again applying Lemma~\ref{lemma:matchings-norm}, one deduces 
\begin{align}
	\|V^{(d-1)}_{S^U}(\alpha_n)\|_2^2 & = \sum_{r=0}^{\floor{\half[d-2]}}\left(\frac{\alpha_n^{2r+1}}{(2r+3)!!}\right)^2\frac{(d-1)!}{(d-2r-2)!}2^{d-4r-3}\sum_{s=0}^r \frac{(2s+ 1)!!}{(2s)!!} \nonumber\\
	& = \frac{2^{d-1}}{3^{2n+1}d(d+1)} \sum_{r=0}^{\floor{\half[d-2]}} 3^{-4rn} \binom{d+1}{2r+3} \label{eq:V_Su_norm}
\end{align}
where in the last equality, we have substituted the value for $\alpha_n$ and used 
\[\sum_{s=0}^r\frac{(2s+1)!!}{(2s)!!}=\frac{(2r+3)!!}{3(2r)!!}.\]

Comparing term by term \eqref{eq:V_idty_norm} with \eqref{eq:V_Su_norm}, one finds that the largest singular value corresponds to $\|V_{\rho}^{(d-1)}(\alpha_n)\|_2$ as long as $n+1\geq \ln(\frac{d-1}{2})/\ln(9)$. Then, arguing similarly to \eqref{eq:RL_bound}, 
\[
\|V_{\rho}^{(d-1)}(\alpha_n)\|_2^2 = 2^{d-1}\sum_{r=0}^{\floor{\half[d-1]}} \frac{3^{-4rn}}{2r+1}\binom{d-1}{2r} \leq 2^{d-1}\left(1+\frac{(1+3^{-2n})^{d-1}-1}{3}\right)
\]
Inserting this into \eqref{eq:EL_norm_general} produces \eqref{eq:EL_infty_norm}.
\end{proof}

\subsection{Proof of Theorem~\ref{thm:explicit_bound}}\label{sec:TheActualProof}
\begin{proof}[Proof of Theorem~\ref{thm:explicit_bound}]
Let $G = (V,E)$ be any simple graph such that $\Delta(G) = \sup_{v\in V} \deg(v) \geq 3$. As discussed in Section~\ref{sec:main_result}, Theorem~\ref{thm:explicit_bound} follows immediately from proving
\[
\epsilon_G(n) := \sup_{(v_L,v_R)\in E} \epsilon_{(v_L,v_R)}(n) < \frac{1}{\Delta(G)}.
\]
Set $d_{\#}=\deg(v_{\#})$ for $\#\in\{L,R\}$. Then Proposition~\ref{prop:espilon-estimate} shows that, as long as the maximum assumption is satisfied, $\epsilon_G(n) \leq \sup_{(v_L,v_R)\in E}\delta_{d_L,d_R}(n)$, where
\[
\delta_{d_L,d_R}(n) =  4 a(n) \left(\frac{1}{\sqrt{(1-b(d_L,n))(1-b(d_R,n))}} + \frac{4 a(n) + b(d_L,n)b(d_R,n)}{(1-b(d_L, n))(1-b(d_R,n))}\right).
\]
and $b(d,n)= 4a(n) \tfrac{\|\bE_L^{(d,n)}\|_\infty}{q_L^{(d,n)}}$ by \eqref{eq:b_reduction}.

Now, define the function
\begin{equation*}
    c(d,n) = 4a(n)\cdot
    \begin{cases}
        \displaystyle\frac{\|\bE_L^{(d,n)}\|_\infty}{q_L^{(d,n)}}, & 1\leq d \leq 4 \\
        \displaystyle2^{d/2-1}\frac{\sqrt{3[2+(1+3^{-2n})^{d-1}]}}{4-(1+3^{-n+1/2})^{d-1}} & d>4
    \end{cases}
\end{equation*}
The constraint $n\geq n(\Delta(G))$ is sufficient to guarantee that the denominator above is strictly positive by Lemma~\ref{lem:quantity_bounds}. It is then easy to check using the values from Table~\ref{tab:small_d} that (for $n$ fixed) this function is increasing in $d$ for $1\leq d\leq \Delta(G)$. Finally, by Lemma~\ref{lem:quantity_bounds}
\[
b(d_{\#}, n) \leq c(d_{\#}, n) \leq c(\Delta(G),n), \qquad \# \in \{L,R\}.
\]
To ease notation, set $\Delta(G) = D$. Two key bounds follow.

The first is that
\begin{align*}
\max\left\{b(d_L,n),\, b(d_R,n),\, \tfrac{b(d_L,n)b(d_R,n)}{4a(n)}\right\} & \leq \frac{3^n}{4}c(D,n)^2 = \frac{2^D}{3^{n-1}}\cdot \frac{2+(1+3^{-2n})^{D-1}}{[4-(1+3^{-n+1/2})^{D-1}]^2}.
\end{align*}
The last fraction in the final expression above is decreasing in $n$. As $n>\ln(2)D/\ln(3)$ for all $D\geq 3$, this is bounded by
\begin{equation}\label{eq:bG_bound}
    \frac{3^n}{4}c(D,n)^2 <  \frac{2^D}{3^{n-1}}\cdot \frac{2+(1+2^{-2D})^{D-1}}{[4-(1+\sqrt{3}\cdot 2^{-D})^{D-1}]^2} = \frac{2^Df(D)}{3^n}
\end{equation}
which is at most one for all $n\geq\frac{\ln(2)D+\ln(f(D))}{\ln(3)}$. This implies that the maximum assumption from Proposition~\ref{prop:espilon-estimate} is satisfied for any edge $(v_L,v_R)\in E$ when $n\geq n(\Delta(G)).$

The second key bound is that
\begin{align}
    \epsilon_G(n) & \leq 4a(n)\left(\frac{1}{(1-c(D,n))} + \frac{4a(n)+c(D,n)^2}{(1-c(D,n))^2}\right) \label{eq:c_delta}\\
    & < \frac{4(3^{n/2}-1)^2}{3^n(3^{n/2}-2)^2} + \frac{28}{3^n(3^{n/2}-2)^2} \nonumber
\end{align}
where the last bound uses $c(D,n) < 2/3^{n/2}$ by \eqref{eq:bG_bound}. The final expression is decreasing in $n$. Since $n> \ln(2)D/\ln(3)$, one finds $\epsilon_n<1/D$ for all $D \geq 5$. In the case of $D=3$ and $D=4$, one can use the values from Table~\ref{tab:small_d} to exactly evaluate \eqref{eq:c_delta}. This yields $\epsilon_G(n) < 1/D$ when $n\geq D$. This completes the proof.
\end{proof}

Let us now discuss the case in which the decoration number is a function of the edge, i.e., when each edge $e$ is decorated with $\bm{n}(e)$ vertices for some $\bm{n}: E \to \mathds{N}$. Let $n=\min_{e\in E} \bm{n}(e)$. We claim that the result of Theorem~\ref{thm:explicit_bound} still holds for the more generalized decorated model so long as $n\geq \Delta(G)$. 

Fix a pair $(v_L,v_R)$ of adjacent sites in $G$ with degrees $d_L$ and $d_R$. We show that $\epsilon_{(v_L,v_R)}(\bm{n}) := \norm{G_{v_L}G_{v_R} - G_{v_L} \wedge G_{v_R}}  \le \delta_{d_L,d_R}(n)$ where $\delta_{d_L, d_R}(n)$ is again the function from Proposition~\ref{prop:espilon-estimate}.

Enumerate the edges that are incident to $v_L$ but not to $v_R$ as $\{e_L^1,\dots, e_L^{d_L-1}\}$, and similarly the ones that are incident to $v_R$ but not to $v_L$ as $\{e_R^1,\dots, e_R^{d_R-1}\}$. Then the transfer operators corresponding to the regions $X_L$ and $X_R$ are respectively
\[
    \bE_L^{\bm{n}} = \bE^{\bm{n}(e_L^1)} \otimes \cdots \otimes \bE^{\bm{n}(e_L^{d_L-1})} \circ \bE_L^{(d_L,0)} = \bE^{\bm{n}(e_L^1)-n} \otimes \cdots \otimes \bE^{\bm{n}(e_L^{d_L-1})-n} \circ \bE_L^{(d_L,n)}
\]
and
\[
    \bE_R^{\bm{n}} = \bE_R^{(d_R,0)}  \circ \bE^{\bm{n}(e_R^1)} \otimes \cdots \otimes \bE^{\bm{n}(e_L^{d_R-1})} =  \bE_R^{(d_R,n)} \circ \bE^{\bm{n}(e_R^1)-n} \otimes \cdots \otimes \bE^{\bm{n}(e_R^{d_R-1})-n}.
\]
From the fact that $\norm{\cdot}_\infty$ is sub-multiplicative and $\norm{\bE}_\infty = 1$, it immediately follows that
\[
\norm{\bE_L^{\bm{n}}}_\infty \le \norm{\bE_L^{(d_L,n)}}_\infty \qc \norm{\bE_R^{\bm{n}}}_\infty \le \norm{\bE_R^{(d_R,n)}}_\infty.
\]
Now let 
\[
    Q_L^{\bm{n}} = \bE_L^{\bm{n}}(\idty) \qc  \quad  Q_R^{\bm{n}} = (\bE_R^{\bm{n}})^*(\rho),
\]
and let $q_L^{\bm{n}}$ and $q_R^{\bm{n}}$ their minimal eigenvalues, respectively. 
Then we have that
\[
    Q_L^{\bm{n}} = \bE^{\bm{n}(e_L^1)-n} \otimes \cdots \otimes \bE^{\bm{n}(e_L^{d_L-1})-n} (Q_L^{(d_L,n)}) \qc \text{with } Q_L^{(d_L,n)} = \bE_L^{(d_L,n)} (\idty).
\]
If $q_L^{(d_L,n)}$ is the minimal eigenvalue of $Q_L^{(d_L,n)}$, then $Q_L^{(d_L,n)} \ge q_L^{(d_L,n)} \idty$ and from the positivity of the transfer operators maps, it follows that
\[
Q_L^{\bm{n}} \ge q_L^{(d_L,n)}\ \bE^{\bm{n}(e_L^1)-n} \otimes \cdots \otimes \bE^{\bm{n}(e_L^{d_L-1})-n}\ (\idty) = q_L^{(d_L,n)}\  \idty,
\]
as $\bE(\idty)=\idty$. Hence, $q_L^{\bm{n}} \ge q_L^{(d_L,n)}$. A similar calculation for $Q_R^{\bm{n}}$ shows that $q_R^{\bm{n}} \ge q_R^{(d_R,n)}$.

The final step is to adapt the proof of Proposition~\ref{prop:espilon-estimate} to this case. Denote by
\[
    b_L(e,\bm{n}) = 4 a(\bm{n}(e)) \frac{ \norm{E_L^{\bm{n}}}_\infty }{q_L^{\bm{n}}} \qc
    b_R(e,\bm{n}) = 2 a(\bm{n}(e)) \frac{ \norm{E_R^{\bm{n}}}_\infty }{q_R^{\bm{n}}},
\]
where $e=(v_L, v_R)$.
From the previous discussions and the fact that $a(n)$ is monotone decreasing, we see that these quantities are not larger than the ones corresponding to the same graph with decoration $n$, namely, $b_L(d_L,n)$ and $b_R(d_R,n)$. From this is follows that 
$\epsilon_{(v_L,v_R)}(\bm{n})$ is upper bounded by $\delta_{d_L, d_R}(n)$.

\begin{acknowledgements}
A.L. was supported by grants PID2020-113523GB-I00 and CEX2019-000904-S, funded by MCIN/AEI/10.13039/501100011033,
by grant RYC2019-026475-I, funded by MCIN/AEI/ 10.13039/501100011033 and ``ESF Investing in your future'', 
and by Comunidad de Madrid (grant QUITEMAD-CM, ref. S2018/TCS-4342). %
A.Y. was supported by the DFG under EXC-2111--390814868. %
The authors acknowledge support of the Erwin Schrödinger International Institute for Mathematics and Physics (ESI), where part of this work was carried out during the ``Tensor Networks: Mathematical Structures and Novel Algorithms'' workshop.
They also thank Bruno Nachtergaele for the helpful discussions during the development of this work, as well as the reviewers whose careful assessments of our work led to improvements in our results and proofs.
\end{acknowledgements}

\section*{Author declarations}
\subsection*{Conflict of Interest}
The authors have no conflicts to disclose.
\subsection*{Author Contributions}
All authors contributed equally to this work.
\section*{Data Availability Statement}
Data sharing is not applicable to this article as no new data were created or analyzed in this study.

\appendix

\section{Generating function calculations}

We now state and prove the result used to calculate $\sum_{k=1}^r h(2r,k)3^k$ in Lemma~\ref{lemma:matchings-norm}.

\begin{lemma}\label{lemma:gf-even-cyles}

Fix $r\geq 1$ and let $h(2r,k)$ denote the number of permutations in $\mathfrak{S}_{2r}^e$ with exactly $k$ cycles.
Then 
\begin{equation}
    \sum_{k=1}^{r} h(2r,k) y^{k} = \qty(\frac{y}{2})^{\bar r} 2^r (2r-1)!!
\end{equation}
where $y^{\bar r} := y (y+1) \cdots (y+r-1)$ denotes the raising factorial.
\end{lemma}
\begin{proof}
We compute this generating function using the exponential formula for labelled combinatorial structures~\cite{generatingfunctionology}.
Let $d_r = (r-1)!$ be the number of cyclic permutations of $[r]$, and set 
\[ 
    \mathcal{D}(x) = \sum_{r \ge 0} d_{2r} \frac{x^{2r}}{(2r)!} = \log( \frac{1}{(1-x^2)^{1/2}}).
\]
This is the \emph{deck enumerator function} of the exponential family of even length cycles. The corresponding \emph{hand enumerator} is then
\begin{equation}\label{eq:hand_function}
    \cH(x,y) = \sum_{r\geq 0}\sum_{k= 0}^{r} h(2r,k)y^k\frac{x^{2r}}{(2r)!}  
\end{equation}
where $h(0,0)=1$ by convention. The exponential formula (see Ref.~\onlinecite[Theorem 3.4.1]{generatingfunctionology}) states that $\cD(x)$ and $\cH(x,y)$ are related via
\[ 
\mathcal{H}(x,y) = e^{y \cD (x)} = (1-x^2)^{-\frac{y}{2}}.
\]
Expanding $(1-x^2)^{-\frac{y}{2}}$ with the generalized binomial series produces
\begin{align*}
\mathcal{H}(x,y) &= \sum_{r\ge 0} \binom{-\frac{y}{2}}{r} (-1)^r x^{2r} = \sum_{r \ge 0}  \frac{y}{2}\qty(\frac{y}{2}+1) \cdots \qty(\frac{y}{2} + r-1) \frac{x^{2r}}{r!} \\ 
& = \sum_{r \ge 0} \qty(\frac{y}{2})^{\bar r} \frac{(2r)!}{r!} \frac{x^{2r}}{(2r)!} = \sum_{r \ge 0} \qty(\frac{y}{2})^{\bar r} 2^r (2r-1)!! \frac{x^{2r}}{(2r)!},
\end{align*}
which when compared with \eqref{eq:hand_function} implies the result.
\end{proof}

\section{Modifications to spectral gap estimates}
\label{sec:norm-comparison}
The aim of this section is twofold. The first is to present the modifications to Ref.~\onlinecite{decorated-aklt} that are needed to obtain Proposition~\ref{prop:espilon-estimate}. The second is to compare the two approaches in the case of the decorated AKLT models, from which we will conclude that the new approach produces a better result for the minimal decoration needed to guarantee a spectral gap for this model. As we only consider decorated lattices in this section, to simplify notation set 
\[\bE_\# = \bE_\#^{(d_\#, n)}, \qquad q_\# = q_\#^{(d_\#,n)}.\]

Let us recall the bound on $\epsilon_n$ from Ref.~\onlinecite[Proposition 3.6]{decorated-aklt} applied to our setting:
\begin{prop}[{Ref.~\onlinecite[Proposition 3.6]{decorated-aklt}}]
\label{prop:old-epsilon-estimate}
    For any edge $(v_L, v_R)$ of a simple graph $G$ with degrees $d_L$ and $d_R$, repsectively, let
    \begin{equation}\label{eq:b_op}
    b^{\rm{op}}_L(n) = 8 a(n) \frac{ \norm{\bE_L}_{\rm{op}}}{q_L} \qc b^{\rm{op}}_R(n) = 4 a(n) \frac{ \norm{\bE_L}_{\rm{op}}}{q_R}, \quad b_{LR}^{\rm op}(n) = \frac{ b^{\rm{op}}_L(n) b^{\rm{op}}_R(n)}{8a(n)}.
    \end{equation}
    If $\max\left\{b_L^{\rm{op}}(n),\, b_R^{\rm{op}}(n),\, b^{\rm{op}}_{LR}(n)\right\} <1$, then $\epsilon^{(v_L,v_R)}_n \leq \delta^{\rm{op}}(n)$ where
\begin{equation}\label{eq:old_epsilon}
\delta^{\rm{op}}(n) = 4 a(n) \left(\frac{1}{\sqrt{(1-b_L^{\rm{op}}(n))(1-b_R^{\rm{op}}(n))}} + \frac{4 a(n)(1 + b_{LR}^{\rm{op}}(n))}{(1-b_L^{\rm{op}}(n))(1-b_R^{\rm{op}}(n))}\right).
\end{equation}
\end{prop}

\subsection{The alternate inner product bound}

The proof of~Proposition~\ref{prop:old-epsilon-estimate} relies on the estimate contained in Ref.~\onlinecite[Lemma~3.3]{decorated-aklt}. Here we prove a variation of that result, from which the estimate of~Proposition~\ref{prop:espilon-estimate} follows.

For each $\Lambda\in\{Y_{v_L}\cup Y_{v_R}, \, Y_{v_L}, \, Y_{v_R}\}$, let $\cK_\Lambda$ be the virtual matrix space for the ground states of $\Lambda$ in the TNS representation, namely:
\[ 
\cK_{Y_{v_L}\cup Y_{v_R}} = M_{2^{d_L-1}\times 2^{d_R-1}}(\bC) \qc \cK_{Y_{v_L}} =   M_{2^{d_L-1}\times 2}(\bC) \qc \cK_{Y_{v_R}} =  M_{2\times 2^{d_R-1}}(\bC). 
\]
These spaces parametrize the ground states in the following sense: there exists linear maps $\Gamma_\Lambda : \cK_\Lambda \to \cH_\Lambda$ such that $\ran \Gamma_\Lambda$ is exactly the subspace of ground states on region $\Lambda$ (for an explicit definition in terms of the tensors of the TNS representation, see Ref.~\onlinecite[eq. (3.4)]{decorated-aklt}).

We define positive semi-definite Hermitian form on $\cK_\Lambda$, denoted by $\langle\cdot,\cdot\rangle_\Lambda$, via
\begin{eqnarray}
\langle B,C\rangle_{Y_{v_L}\cup Y_{v_R}} & = & \Tr (Q_R B^* Q_L C) \label{eq:IP_vw}\\
\langle B, C \rangle_{Y_{v_L}} & = & \Tr (\rho B^* Q_L C) \label{eq:IP_v}\\
\langle B, C \rangle_{Y_{v_R}} & = & \Tr (Q_R B^* C)\label{eq:IP_w}
\end{eqnarray}
These are inner products as long as $Q_L=\bE_L(\idty)$ and $Q_R=(\bE_R)^*(\rho)$ are (strictly) positive-definite, which by Lemma~\ref{lem:quantity_bounds} holds for the decorated AKLT models when $n\geq \frac{\ln(d-1)}{\ln(3)}-\frac{\ln(\ln(3))}{\ln(3)}+\frac{1}{2}$. In this case, one can also verify that the maps $\Gamma_\Lambda$ are injective. They also satisfy the following approximation bound, which is a variant of Ref.~\onlinecite[Lemma~3.3]{decorated-aklt}.
\begin{lemma} \label{lem:IPEstimates}
Let $\Lambda\in\{Y_{v_L}\cup Y_{v_R}, \, Y_{v_L}, \, Y_{v_R}\}$. Then for any $B, C \in \cK_\Lambda$,
\begin{equation}\label{innerLambda}
\left| \langle \Gamma_\Lambda(B), \Gamma_\Lambda(C) \rangle - \langle B, C \rangle_\Lambda\right|  \leq  2a(n) C_\Lambda \|B\|_2\|C\|_2,
\end{equation}
where the constants are defined by
\begin{equation}
C_{Y_{v_L}\cup Y_{v_R}} = \|\E_L\|_\infty\|\E_R\|_\infty,  \qquad C_{Y_{v_L}} = \|\E_L\|_\infty,  \qquad C_{Y_{v_R}}=\|\E_R\|_\infty.
\end{equation}
\end{lemma}

We discuss the differences between this and Ref.~\onlinecite[Lemma~3.3]{decorated-aklt} and then give the proof. First, since the virtual bound dimension is $D=2$ for the decorated AKLT models, the constant prefactor in the bound from Ref.~\onlinecite[Lemma~3.3]{decorated-aklt} is obtained from replacing $C_\Lambda$ with
\[
C_{Y_{v_L}\cup Y_{v_R}}' = 2\|\E_L\|_{\rm op}\|\E_R\|_{\rm op}, \quad C_{Y_{v_L}}' = 2\|\E_L\|_{\rm op}, \quad \text{and}\quad C_{Y_{v_R}}'=2\|\E_R\|_{\rm op},
\]
where $\|\cdot\|_{\rm op}$ denotes the norm induced from the operator norm from \eqref{eq:op_norm}.

 The other difference is that the bound in Ref.~\onlinecite[Lemma~3.3]{decorated-aklt} is given in terms of the operator norm of $B$ and $C$ instead of the Hilbert-Schmidt norm. However, in the previous work, the operator norm is immediately bounded from above using one of the norms induced by \eqref{eq:IP_vw}-\eqref{eq:IP_w}. The Hilbert-Schmidt norm satisfies the same bound: 
\begin{equation}\label{norm_equivs}
\|B\|_2  \leq  \frac{1}{\sqrt{q_Lq_R}}\|B\|_{Y_{v_L}\cup Y_{v_R}}, \quad
\|B\|_2  \leq  \frac{1}{\sqrt{\rho_{\rm min}q_L}}\|B\|_{Y_{v_L}}, \quad
\|B\|_2  \leq  \frac{1}{\sqrt{q_R}}\|B\|_{Y_{v_R}}
\end{equation}
where $B\in\cK_\Lambda$ for the appropriate choice of $\Lambda \in\{Y_{v_L}\cup Y_{v_R}, \, Y_{v_L}, \, Y_{v_R}\}$. The above should be compared with Ref.~\onlinecite[eq. (3.28)]{decorated-aklt}. Here, $\rho_{\rm min} := \min \spec (\rho) = 1/2$.

Using instead Lemma~\ref{lem:IPEstimates} and \eqref{norm_equivs} in Ref.~\onlinecite{decorated-aklt}, all arguments run as stated with the small modification of replacing $C_\Lambda'$ with $C_\Lambda$. This results in Proposition~\ref{prop:espilon-estimate} from Section~\ref{sec:transfer_ops}.

\begin{proof}
Similar to the proof of Ref.~\onlinecite[Lemma 3.3]{decorated-aklt}, we prove the result for $\Lambda = Y_{v_L}\cup Y_{v_R}$ as the other two cases follow from simple modifications of this case. 

Let $\cB = \{\ket{0}, \ket{1}\}$ be an orthonormal basis of $\bC^2$. Then,
\begin{eqnarray}
\langle \Gamma_{\Lambda}(B), \Gamma_{\Lambda}(C) \rangle = 
\sum_{\alpha, \beta = 0,1} \bra{\alpha} \E^{n}\circ \E_R\big[B^*\E_L(\ketbra{\alpha}{\beta})C\big] \ket{\beta}, \label{IP_gen1}
\end{eqnarray}
(see Ref.~\onlinecite[eq. (3.25)]{decorated-aklt}).
Furthermore, the identity
\begin{equation}\label{IP_gen2}
\langle B,C\rangle_\Lambda = \sum_{\alpha, \beta=0,1} \bra{\alpha} \ketbra{\idty}{\rho}\circ \E_R[B^*\E_L(\ketbra{\alpha}{\beta})C\big] \ket{\beta}
\end{equation}
can easily be seen from simplifying the RHS. 

Substituting \eqref{IP_gen1}-\eqref{IP_gen2} into the LHS of \eqref{innerLambda}, and recalling that $E_{\alpha, \beta} = \ketbra{\alpha}{\beta}$ is an orthonormal basis with respect to the Hilbert-Schmidt norm, one finds
\begin{eqnarray}
\left| \langle \Gamma_\Lambda(B), \Gamma_\Lambda(C) \rangle - \langle B, C \rangle_\Lambda\right|
& = &
\left|\sum_{\alpha, \beta = 0,1}  \left\langle E_{\alpha,\beta}, \, (\E^{n}-\ketbra{\idty}{\rho})\circ\E_R\circ S_{B,C}\circ\E_L(E_{\alpha,\beta})\right\rangle_2\right| \nonumber \\
& \leq &
\|(\E^{n}-\ketbra{\idty}{\rho})\circ\E_R\circ S_{B,C}\circ\E_L\|_1 \nonumber\\
& \leq & \|\E^{n}-\ketbra{\idty}{\rho}\|_2\|\E_R\circ S_{B,C}\circ\E_L\|_2 \label{eq:matrix_ip}
\end{eqnarray}
where we used cyclicity of the trace and the Cauchy-Schwarz inequality, denoted by $\|\cdot\|_p$ the Schatten $p$-norm, and introduced the map $S_{B,C}:  M_{2^{d_L-1}}(\bC) \to  M_{2^{d_R-1}}(\bC)$ defined by
$S_{B,C}(A) := B^*AC.$

Recalling the diagonalization of $\bE:M_2(\bC)\to M_2(\bC)$ from \eqref{eq:Spin1_transfer_op}, one finds
\[
\|\E^{n}-\ketbra{\idty}{\rho}\|_2 \leq \sqrt{\dim(M_2(\bC))} \|\E^{n}-\ketbra{\idty}{\rho}\|_{\infty} = 2a(n).
\]
Then, repeatedly applying the generalized H\"older inequality $\|RT\|_2 \leq \|R\|_\infty\|T\|_2$ to \eqref{eq:matrix_ip} produces
\begin{eqnarray}
\left| \langle \Gamma_\Lambda(B), \Gamma_\Lambda(C) \rangle - \langle B, C \rangle_\Lambda\right| & \leq & 2a(n) \|\bE_L\|_\infty \|\bE_R\|_\infty \|S_{B,C}\|_2 \, . \label{G_final_est}
\end{eqnarray}
The final result is then a consequence of substituting $\|S_{B,C}\|_2=\|B\|_2\|C\|_2$, which holds from calculating
\[
\|S_{B,C}\|_2^2 = \sum_{\gamma,\delta=1}^{2^{d_L-1}}\Big\langle B^*\ketbra{\gamma}{\delta}C, \, B^*\ketbra{\gamma}{\delta}C\Big\rangle_2 = \|B\|_2^2\|C\|_2^2
\]
where $\{\ket{\gamma}: 1\leq \gamma \leq 2^{d_L-1} \}$ denotes any orthonormal basis of $\bC^{2^{d_L-1}}$.
\end{proof}

\subsection{Comparing the two bounds for decorated AKLT models}
To compare the two approaches, it is sufficient to assume that $G$ is a simple, regular graph, i.e. $d = \deg(v)$ for all vertices. Thus, recalling  Proposition~\ref{prop:espilon-estimate} and Proposition~\ref{prop:old-epsilon-estimate}, the two approaches show that the AKLT model on $G^{(n)}$ is uniformly gapped if
\begin{equation}\label{eq:gap_conditions}
    \max\left\{b_L^\#(n), \, b_R^{\#}(n), \, b_{LR}^\#(n)\right\}<1 \quad \text{and} \quad \delta^\#(n)<1/d 
\end{equation}
where for the approach of the present work, we redefine the quantities from Proposition~\ref{prop:espilon-estimate} as $\delta^\infty(n)=\delta_{d,d}(n)$, and
\begin{equation}\label{eq:newB}
b_L^{\rm \infty}(n) = 4a(n)\frac{\|\bE_L\|_{\infty}}{q_L}, \quad b_R^{\infty}(n) = 2a(n)\frac{\|\bE_R\|_{\infty}}{q_R}, \quad b_{LR}^{\infty}(n) = \frac{b_L^{\rm \infty}(n) b_R^{\rm \infty}(n) }{4a(n)}.
\end{equation}
We prove that, in the case of decorated AKLT models, the conditions from \eqref{eq:gap_conditions} are necessarily satisfied for the approach used in the present work if they are satisfied for the approach from Ref.~\onlinecite{decorated-aklt}.

\begin{prop}\label{prop:bound_comparison} Suppose that $G=(V,E)$ is a regular, simple graph such that $\deg(v) = d$ for all $v\in V$. Then, for any $n\geq \frac{\ln(d-1)}{\ln(3)}-\frac{\ln(\ln(3))}{\ln(3)}+\frac{1}{2}$, one has $b_{LR}^{\infty}(n)< b_{LR}^{\rm op}(n)$ and
\begin{equation}\label{eq:b_comparison}
    \sqrt{(1-b_L^{\rm op}(n))(1-b_R^{\rm op}(n))} < \sqrt{(1-b_L^{\infty}(n))(1-b_R^{\infty}(n)).}
\end{equation}
Said differently, $ \delta^\infty(n) <\delta^{\rm op}(n)$.
\end{prop}

The proof will make use of the following function, which is even and strictly increasing for $x\geq 0$:
    \begin{equation}\label{eq:F}
          F_{d-1}(x) := \sum_{r=0}^{\floor{\half[d-1]}} \binom{d-1}{2r} \frac{x^{2r}}{2r+1}.
    \end{equation}

\begin{proof}[Proof of Proposition~\ref{prop:bound_comparison}]
To begin, we produce formulas of the necessary quantities for the comparison. Recall that $\alpha_n = 2(-3)^{-n}$. Then by \eqref{eq:ER},
\begin{equation}\label{eq:ELR_infty_norm}
    \|\bE_L\|_\infty= \|\bE_R\|_{\infty}=\frac{d+1}{2^{d-1/2}}\|V_{\rho}^{(d-1)}(\alpha_n)\|_2 = \frac{d+1}{2^{d/2}}\sqrt{F_{d-1}(\alpha_n^2/4)}.
\end{equation}
Here, we have invoked Lemma~\ref{lem:quantity_bounds}, and used \eqref{eq:V_idty_norm} and \eqref{eq:F} for the last equality. 

On the other hand, the norm induced by the operator norm satisfies $\|\bE_{\#}\|_{\rm op} = \|\bE_{\#}(\idty)\|$ since $\bE_{\#}$ is a completely positive map. Therefore, considering \eqref{eq:undecorated_transfer_op},
\begin{equation}\label{eq:ER_op_norm}
    \bE_R(\idty) =  \sum_{k=0}^d W_k^d(W_k^d)^* = \frac{d+1}{2}\idty \;\; \implies \;\;  \|\bE_{R}\|_{\rm op}=\frac{d+1}{2},
\end{equation}
while by Lemma~\ref{lem:SVD_Decomp},
\begin{equation}\label{eq:EL_op_norm}
   \|\bE_{L}\|_{\rm op} = \frac{d+1}{2^{d-1}}\|V_{\rho}^{(d-1)}(\alpha_n)\| \geq \frac{d+1}{2^{d-1}}F_{d-1}(|\alpha_n|/2).
\end{equation}
The lower bound in \eqref{eq:EL_op_norm} follows from using the definition of $V_{\rho}^{(d-1)}(\alpha_n)$ and \eqref{eq:right-largest-eigenvector} to calculate
\begin{align*}
	\|V_{\rho}^{(d-1)}(\alpha_n)\| & \ge \expval{V_{\rho}^{(d-1)}(\alpha_n)}{\uparrow\cdots \uparrow} 
	=   F_{d-1}\qty(|\alpha_n|/2).
\end{align*}

We can now compare the desired quantities. Considering \eqref{eq:b_op} and \eqref{eq:newB} it is clear that $b_{LR}^{\infty}(n)< b_{LR}^{\rm op}(n)$ if and only if 
$\|\bE_L\|_{\infty}\|\bE_R\|_{\infty} < 2\|\bE_L\|_{\rm op}\|\bE_R\|_{\rm op}.$ From \eqref{eq:ELR_infty_norm}-\eqref{eq:EL_op_norm}, it is easily deduced that this inequality holds since 
\[F_{d-1}(\alpha_n^2/4)= F_{d-1}(9^{-n})< 2F_{d-1}(3^{-n})=2F_{d-1}\qty(|\alpha_n|/2).\]

Now consider \eqref{eq:b_comparison}, and note that \eqref{eq:ER} implies $Q_R = \frac{1}{2}Q_L$ and $b_L^\infty(n) = b_R^\infty(n)$. Thus, by \eqref{eq:newB} and \eqref{eq:ELR_infty_norm},
\[
\sqrt{(1-b^{\infty}_L(n))(1-b^{\infty}_R(n))} = 1-b^{\infty}_L(n) = 1-\frac{ a(n)(d+1)}{2^{d-3}q_L} \cdot \sqrt{2^{d-2}F_{d-1}(\alpha_n^2/4)}.
\]
On the other hand, applying \eqref{eq:b_op} and \eqref{eq:ER_op_norm}-\eqref{eq:EL_op_norm}
\begin{align*}
	\sqrt{(1-b^{\rm op}_L(n))(1-b^{\rm op}_R(n))} & \le 1-\frac{b^{\rm op}_L(n)+b^{\rm op}_R(n)}{2} \le 1- \frac{ a(n)(d+1)}{2^{d-3}q_L}\qty( F_{d-1}\qty(\abs{\alpha_n}/2) + 2^{d-2}) .
\end{align*}
As $F_{d-1}\qty(\abs{\alpha_n}/2)>0$, \eqref{eq:b_comparison} is a consequence of using $\binom{n}{k} = \binom{n-1}{k}+ \binom{n-1}{k-1}$ for $1\leq k \leq n-1$ to bound
\[
F_{d-1}(\alpha_n^2/4)< F_{d-1}(1) < \sum_{r=0}^{\floor{\tfrac{d-1}{2}}}\binom{d-1}{2r} = \sum_{r=0}^{d-2}\binom{d-2}{r} = 2^{d-2}.
\]
\end{proof}

\bibliography{bibliography}

\end{document}